\documentclass[11pt]{article}
\usepackage[utf8]{inputenc}
\usepackage{tikz}
\usepackage{caption}
\usetikzlibrary{shapes,backgrounds}
\usepackage[margin=1in]{geometry}
\usepackage{setspace}
\usepackage{subcaption}
\usepackage{wrapfig}
\usepackage{comment}
\usepackage{enumitem}
\usepackage{caption}
\usetikzlibrary{calc}
\usepackage{todonotes}
\usepackage{amsmath,amssymb,amsthm,graphicx,caption}
\usepackage{fullpage}
\usepackage{hyperref}
\usepackage{algorithm}
\usepackage{algorithmic}
\usepackage{todonotes}
\usepackage{times}

\newtheorem{Claim}{Claim}

\newtheorem{theorem}{Theorem}
\newtheorem{definition}[]{Definition}
\newtheorem{corollary}[]{Corollary}

\newtheorem{remark}{Remark}
\newtheorem{conjecture}{Conjecture}

\newtheorem{lemma}{Lemma}

\newcommand{\proportional}{\textsc{cds-clearing}}

\newcommand{\pure}{\textsc{pure-circuit}}

\allowdisplaybreaks  
\usetikzlibrary{positioning} 
\usetikzlibrary{automata} 
\usetikzlibrary[automata]
\usetikzlibrary{snakes}

\title{Clearing Financial Networks with Derivatives: From    Intractability to Algorithms}

\author{Stavros D. Ioannidis \and Bart de Keijzer \and Carmine Ventre}
\date{}
\date{Department of Informatics\\King's College London}

\begin{document}

\maketitle
\thispagestyle{empty}
\begin{abstract}
Financial networks raise a significant computational challenge in identifying insolvent firms and evaluating their exposure to systemic risk. This task, known as the \emph{clearing problem}, is computationally tractable when dealing with simple debt contracts \cite{eisenberg2001systemic}. However under the presence of certain derivatives called credit default swaps (CDSes) \cite{schuldenzucker2016clearing} the \emph{clearing problem} is $\textsf{FIXP}$-complete \cite{ioannidis2022strong}. Existing techniques only show $\textsf{PPAD}$-hardness for finding an \emph{$\epsilon$-solution} for the \emph{clearing problem} with CDSes within an unspecified small range for $\epsilon$ \cite{schuldenzucker2017finding}.

We present significant progress in both facets of the \emph{clearing problem}: (i) intractability of approximate solutions; (ii) algorithms and heuristics for computable solutions. Leveraging \pure\ \cite{DBLP:conf/focs/DeligkasFHM22}, we provide the first explicit inapproximability bound for the \emph{clearing problem} involving CDSes. 
Our primal contribution is a reduction from \pure\ which establishes that finding approximate solutions is \textsf{PPAD}-hard within a range of roughly 5\%.

To alleviate the complexity of the \emph{clearing problem}, we identify two meaningful restrictions of the class of financial networks motivated by regulations: (i) the presence of a central clearing authority; and, (ii) the restriction to \emph{covered} CDSes. We provide the following results:
\begin{enumerate}
\item[$i$)] The \textsf{PPAD}-hardness of approximation persists when central clearing authorities are introduced.
\item[$ii$)] An optimisation-based method for solving the \emph{clearing problem} with central clearing authorities.
\item[$iii$)] A polynomial-time algorithm when the two restrictions hold simultaneously. 
\end{enumerate}
\end{abstract}
\newpage
\setcounter{page}{1}
\section{Introduction}
How does the collapse of an established financial institution affect the stability of a financial market? Can a bank's exposure to such a contagion risk be efficiently measured by clearing authorities? What factors should legislatures consider when formulating economic policies? These are just some exemplar questions that are paramount to the study of systemic risk in finance. In this context, financial networks have emerged as the framework of reference. These networks are modelled as graphs whose nodes represent financial institutions (such as banks) and edges represent active economic commitment between nodes. The central computational challenge for financial networks, that underpins many, if not all, the questions considered in the literature in this area, is the \emph{clearing problem}. This amounts to computing the percentage of the liabilities that each bank should pay if the system had to be cleared at once -- intuitively, the smaller the ratio, the more exposed the bank is to systemic risk due to defaults in the network. In the literature, this ratio is referred to as the ``\emph{clearing recovery rate}."

The finance industry is infamous for its 
inventiveness in coming up with hedging and/or lucrative new instruments that \emph{derive} their value from  underlying assets. Chief amongst these derivative contracts is the so-called Credit Default Swap (CDS) for which the value of the underlying asset is based upon the recovery rate of a third institution, called reference bank. Originally conceived to protect the creditor from potential default of its debtor, CDSes have been increasingly used as speculative contracts to wager about the solvency of other institutions. They attracted the interest of regulators both in the US and the EU being amongst the causes of the Great Financial Crisis of 2007 and the European sovereign debt crisis of 2011. The picture for the computational complexity of the clearing problem is less clean when CDSes enter the frame. 

In their seminal work Eisenberg and Noe \cite{eisenberg2001systemic}, establish the mathematical formulation of the \emph{clearing problem}. Their approach, which constitutes the baseline of research in this area, initiates the study of the problem from a fixed point computation perspective. The upshot of their work, is that \emph{clearing recovery rates} correspond to  fixed points of a certain function $f$,  whenever the liabilities are simple contracts requiring the payment of a certain value (known as notional) unconditionally. More importantly, they provide a polynomial time algorithm for computing the clearing recovery rate vectors in case of networks with simple debt contracts. Despite the solid foundational base of this work in the financial literature, the reality that emerges from the presence of credit default swaps cannot be captured.

The addition of credit default swaps to the model is attributed to Schuldenzucker, Seuken and Battinston \cite{schuldenzucker2016clearing}. The authors present a graph-like model for financial networks that contain simple debt and CDS contracts. Similarly, the \emph{clearing problem} takes the form of a fixed point computation over a more complex function $f$ that always admits fixed points. The main observation for the new model on top of its involved combinatorial structure, is the numerical irrationality of the clearing recovery rates; a combination that constitutes the exact solution of the problem out of reach. To correctly define the computational question in this case, one then needs to resort to a notion of approximation.

Proposed approximation notions in the literature involve the computation of either a `weak' (almost) fixed point $x$ such that $|x -f(x)|< \epsilon$ along all dimensions or a `strong' (near) fixed point $y$ satisfying $|y-z|<\epsilon$ where $z$ is such that $z=f(z)$ (that is, $z$ is a fixed point). Both versions can be appealing for systemic risk in financial networks and have been considered in the literature: Ioannidis, de Keijzer and Ventre showed that the exact version of the clearing problem involving CDSes is \textsf{FIXP}-complete, which immediately yields \textsf{FIXP}$_a$-completeness for strong approximations whenever the banks pay proportionally or by prioritising debts \cite{ioannidis2022strong,ioannidis2022financial}.

In \cite{schuldenzucker2017finding} it  is proved that there a exit a constant $\epsilon$, for which computing  weak approximations within this constant is \textsf{PPAD}-complete, whenever the banks pay proportionally. The authors construct a reduction from the 
$\epsilon$-\emph{Generalised-Circuit} problem proved by Rubinstein to be \textsf{PPAD}-hard for some constant $\epsilon$ \cite{rubinstein2015inapproximability}. This statement excludes the possibility of constructing a polynomial time approximation scheme (\textsf{PTAS}).

\vspace{2mm}
\noindent {\bf Contribution.}
In this paper we want to understand whether there are values of $\epsilon$ for which computing a suitable approximation of the clearing recovery rates can be done in polynomial time. In particular, we are interested in 
the intractability of the problem; for which values of $\epsilon$ do the hardness results carry over? We study  weak approximations and establish the following result.

\begin{center}
\begin{minipage}[c]{0.95\textwidth}
    \textbf{Main Theorem 1 (informal).} Computing $\epsilon$-weak approximations of clearing recovery rates is \textsf{PPAD}-hard for $\epsilon \leq \frac{3-\sqrt{5}}{16} \approx 0.048$.
    \end{minipage}
    \end{center}


We exploit the connection of the clearing problem to the \textsf{PPAD} complexity class, along with the recent advancements in the area of  total function search problems \textsf{TFNP}, regarding constant inapproximability. Our proof leverages \emph{Pure-Circuit} -- a toolkit introduced by Deligkas, Fearnley, Hollender and Melissourgos \cite{DBLP:conf/focs/DeligkasFHM22} -- for showing constant inapproximability in \textsf{PPAD}. The problem features a set of variables which can take one of three different values 0, 1 or garbage and a circuit constructed from connecting gates of the following three types: NOT, OR and PURIFY. The first two basically follow boolean logic whereas the PURIFY gate has one input and two outputs and makes sure to duplicate the pure bit in input or produce at least one pure bit in output if the input is garbage. The main result of \cite{DBLP:conf/focs/DeligkasFHM22} is that \emph{Pure-Circuit} is \textsf{PPAD}-complete.

We construct a direct reduction from \emph{Pure-Circuit} to the \emph{clearing problem}. The proof is technically involved and can be divided in the following parts: (i) A ``hypothetical'' encoding of the recovery rate space $[0,1]$ into three disjoint subintervals that intuitively map recovery rates to either a pure bit or garbage. (ii) The construction of three ``financial gates'' whose behaviour under weak approximate clearing vectors must simulate the function of \emph{Pure-Circuit} gates after decoding the recovery rates of specific banks. (iii) A reverse engineering process that fixes the initial encoding for the recovery rate values and results in the final specified constant $\epsilon$. Here we note two things. The proposed reduction is optimal, in the sense that it generates a family of possible encodings for the recovery rate values, upon which we choose the one that maximises the inapproximability parameter $\epsilon$. The intricate structure of the proposed networks combined with the restriction of computations to intervals within the $[0,1]$ range, significantly increases the complexity of the analysis. To bypass this obstacle we define a first-order language denoted as $L(\text{R},\text{F},\text{C})$, designed to handle combined arithmetic operations involving intervals and numbers within the range $[0,1]$. This is a basic tool in our analysis since it combines concise representations with a high level of expressiveness.
The upshot is that if we were able to approximate the clearing problem better than (roughly) $0.048$ then the encoding would guarantee that we could solve \emph{Pure-Circuit}. 
%

In the second part we evaluate the effectiveness of policies introduced in the wake of the aforementioned crises, the introduction of Central CDS Debtors (CCDs), a construct reminiscent of a \emph{central clearing counterparty} (CCP), and the \emph{ban of naked CDSes}. 

\begin{center}
\begin{minipage}[c]{0.95\textwidth}
    \textbf{Main Theorem 2 (informal).} There is an exponential-time algorithm for optimising any objective function of interest with respect to the clearing recovery rates if CCDs are mandated. 
    \end{minipage}
    \end{center}

Both the regulatory frameworks in EU and US 
require the use of a CCP for a large part of the over-the-counter derivatives market \cite{schuldenzucker2016clearing}. This means that if two banks $u$ and $v$ want to sign a derivative (CDS in our context) they need to involve a CCP which will sign one CDS contract with $u$ and another CDS with $v$. A CCP is typically 
so well capitalised that it can absorb shocks in the market. A Central CDS debtor is a CCP which is not creditor of any CDS contract, a strengthening of the current notion of CCP where the risk of the original CDS from $u$ to $v$ is not only absorbed downstream towards $v$ but also upstream vis-a-vis $u$ (e.g., the original CDS is only substituted with a CDS between the CCP and $v$ whereas the liability with $u$ is a simple debt contract). Such a policy requirement restricts quite significantly the network structure and holds out the hope for efficient computation of (good approximations of) clearing rates. Interestingly, under the presences of CCDs, all clearing vectors are rational. Unfortunately, our inapproximability result applies also to the case in which CCDs ought to be used.

We propose an approach for computing exact clearing vectors. 
Specifically, for financial networks that involve CCDs, we devise a Mixed-Binary Linear Program.  Within this program, there is a combination of real-valued variables for the recovery rates and binary decision variables tailored to indicate a bank's solvency status. The feasible solutions of this program correspond to the clearing recovery rate vectors of the network. This framework is highly adaptable for optimising any linear objective function of interest related to the clearing vector. An immediate implication is an exponential-time algorithm for computing clearing recovery rate vectors when CCDs are mandated. To the best of our knowledge, this is the first explicit algorithm for the clearing problem with derivatives.

\begin{center}
\begin{minipage}[c]{0.95\textwidth}
    \textbf{Main Theorem 3 (informal).} There is a polynomial-time algorithm computing exact clearing recovery rates if CCDs are mandated and naked CDSes are banned. 
    \end{minipage}
    \end{center}
    
A naked CDS is a purely speculative contract; 
its counterparties do not have any other interest in the reference bank if not the CDS itself. However, regulators could ban their existence and only allow to buy a CDS if a corresponding (debt) exposure exists -- 
i.e., to only have so-called covered CDS. We exploit the topological structure of covered CDS and CCDs and compute exact clearing recovery rates for any such financial network in polynomial time. 

\vspace{2mm}
\noindent \textbf{Significance of our results.} 
The \emph{clearing problem} is among the first conjectured by the authors in \cite{DBLP:conf/focs/DeligkasFHM22}, to fall into the category of problems for which a transition from an undefined small constant inapproximability parameter to an explicit higher bound can be achieved via their proposed tool, i.e., \emph{Pure-Circuit}. Following the current of research in the intersection of finance and computation and aligning it with the cutting-edge research direction marked by \cite{DBLP:conf/focs/DeligkasFHM22}, we prove the conjecture correct. We provide the first explicit inapproximability bound for the \emph{clearing problem} for financial networks with credit default swaps. The significance of our result is not restricted to specifying the inapproximability bound, but also to its comparison with the state of the art. We specify the initial (yet unknown) inapproximability bound to $\epsilon < 1/150 \approx 0.0067$ (cf. Appendix \ref{apx:C}). Therefore, our main theorem \ref{main_theorem} is a seven-fold improvement over the state of the art. Another important added value is that our result holds for any ``reasonable'' \emph{payment scheme} (see Corollary \ref{corollary1} below). 
 
This paper is the pivotal step from complexity to algorithmic results for the \emph{clearing problem}. Our optimisation approach is novel, and provides a versatile tool for proving higher lower bounds and upper bounds. The proposed algorithms are the first positive results in the area. The exponential algorithm suggests a clearing mechanism for central clearing authorities. Moreover, it solves default ambiguity problems that were previously noted in the literature \cite{papp2022default} (cf. Remark \ref{expremark} below for details). It is important to emphasise that the exponent in the algorithm's running time is contingent solely on the number of banks in the network and remains unaffected by the actual size of the parameters within the networks structure (see Remark \ref{expremark} below). The efficient algorithm we propose, when high-capitalised clearing authorities are coupled with \emph{covered CDSes}, aligns with legislative policies on prohibiting the use of ``naked CDSes." Our analytical proof supports with evidence policies by the EU for sovereign debt in the wake of the Eurozone crisis, (see \cite{EUBusiness}) and previous research \cite{ioannidis2022strong,schuldenzucker2017finding}.

\vspace{2mm}
\noindent \textbf{Related Work.} \label{sec:related} The \emph{clearing problem} was introduced by Eisenberg and Noe \cite{eisenberg2001systemic}. They introduced a model addressed to financial networks with simple debt contracts. The authors presented the first polynomial-time algorithm for computing exact clearing recovery rate vectors whenever banks pay proportionally. Other payment policies and variations of the model were defined by Rogers and Veraart \cite{rogers2013failure}. Concepts of financial systems with simple debt contracts were investigated in \cite{schuldenzucker2021portfolio,papp2021debt,DBLP:conf/sagt/HoeferW22}

Schuldenzucker, Seuken and Battiston \cite{schuldenzucker2016clearing}, added credit default swaps to the model. In \cite{schuldenzucker2017finding}, the authors showed that finding a weak approximate clearing vector in financial networks containing both debt contracts and CDSes is \textsf{PPAD}-complete. Ioannidis, de Keijzer and Ventre \cite{ioannidis2022strong,ioannidis2022financial,DBLP:journals/tcs/IoannidisKV23} showed that computing an exact solution to the  clearing problem with CDSes whenever the banks pay proportionally or under priorities is \textsf{FIXP}-complete. Further studies on financial networks with CDSes are presented by Papp and Wattenhofer \cite{papp2022default,DBLP:conf/icalp/PappW20a,DBLP:conf/innovations/PappW21}.

A game-theoretic approach to financial networks is introduced by Bertschinger, Hoefer and  Schmand \cite{bertschinger2020strategic}. Financial Network games and strategic behaviour of financial firms is presented by Kanellopoulos, Kyropoulou and Zhou  \cite{kanellopoulos2021financial,DBLP:conf/atal/KanellopoulosKZ22}.The strategic framework of prioritising debts in financial networks with CDSes is due to Papp and Wattenhofer \cite{DBLP:conf/icalp/PappW20a}.

The complexity of total search function problems \textsf{TFNP} was first considered by Megiddo and Papadimitriou \cite{DBLP:journals/tcs/MegiddoP91} while the class \textsf{PPAD} was defined by Papadimitriou~\cite{papadimitriou1994complexity,yannakakis2009equilibria} and gained significant attention from the papers of Daskalakis, Goldberg and Papadimitriou \cite{daskalakis2009complexity} and Chen, Deng and Teng \cite{DBLP:journals/jacm/ChenDT09} for computing Nash equilibrium in strategic-formed games. The \textsf{FIXP} complexity class was introduced by Etessami and Yannakakis \cite{etessami2010complexity,yannakakis2009equilibria} for studying strong approximation to Nash equilibrium in strategic-formed games. Recent advancements related to the \textsf{FIXP} class can be found in \cite{DBLP:conf/icalp/BatziouHH21,DBLP:journals/siamcomp/FilosRatsikasGHLP23,DBLP:conf/focs/Filos-RatsikasH21,DBLP:journals/jcss/GoldbergH21,ioannidis2022strong,DBLP:journals/tcs/IoannidisKV23,hansen2021computational}. \emph{Generalised-Circuit} was introduced by \cite{DBLP:journals/jacm/ChenDT09} and used by Rubinstein \cite{rubinstein2015inapproximability} to establish a first constant inapproximability method. Most recently Deligkas, Fearnley, Hollender and Melissourgos \cite{DBLP:conf/focs/DeligkasFHM22} suggested \emph{Pure-Circuit}, a tool for showing stronger constant inapproximability results for \textsf{PPAD}. Applications of the problem are presented in \cite{DBLP:conf/focs/DeligkasFHM22,deligkas2023tight}.
\begin{section}{Preliminaries}
{\bf Financial Networks.} A financial network consists of a set of financial entities (which we refer to as \emph{banks} for convenience), interconnected through a set of financial contracts.
 Let $N = \{1,\ldots,n\}$ be the set of $n$ banks. 
Each bank $i \in N$ has a certain amount of \emph{external assets}, denoted by $e_i \in \mathbb{Q}_{\geq 0}$. We let $e = (e_1, \ldots ,e_n)$ be the vector of all external assets. We consider two types of liabilities among banks:  \emph{debt contracts} and \emph{credit default swaps (CDSes)}. A {debt contract} requires one bank $i$ (the debtor) to pay another bank $j$ (the creditor) 
a certain amount $c_{i,j} \in \mathbb{Q}_{\geq 0}$.
A CDS requires a debtor $i$ to pay a creditor $j$ on condition that a third bank called the \emph{reference bank} $R$ is in default, meaning that $R$ cannot 
fully pay its liabilities. 

We associate with each bank $i$ a variable $r_i \in [0,1]$, called \emph{recovery rate}, which represents the proportion of liabilities it is able to pay. Having $r_i = 1$ means that bank $i$ can fully pay its liabilities, while $r_i < 1$ indicates that $i$ is in default and can only pay off a fraction of $r_i$ of its liabilities. In a CDS with debtor $i$, creditor $j$ and reference bank $R$, bank $i$ is obliged to pay bank $j$ an amount of $(1-r_R)c_{i,j}^R$, where $c_{i,j}^R \in \mathbb{Q}_{\geq 0}$ is the face value of the CDS. We do not allow any bank to have a debt contract with itself, and assume that all three banks in any CDS are distinct. The value $c_{i,j}$ ($c_{i,j}^R$, respectively) of a debt contract (CDS, respectively) is also called the \emph{notional} of the respective contract. Finally, we let $c = (c_{i,j},c_{i,j}^k)_{i,j,k \in N}$ be the collection of all contracts' notionals, in the form of a vector. Note that, here, we identify the absence of a contract as a contract with notional $0$.

The financial network $\mathcal{F}$ is represented as the triplet $(N,e,c)$.\footnote{We let $\mathcal{DC}_\mathcal{F}$ denote the set of all pairs of banks participating in a debt contract in $\mathcal{F}$. Similarly, $\mathcal{CDS}_\mathcal{F}$ denotes the set of all triplets participating in a CDS in $\mathcal{F}$. We drop $\mathcal{F}$ from the notation when it is clear from context.} The \emph{contract graph} of $\mathcal{F}=(N,e,c)$ is defined as a coloured directed multi-graph $G_{\mathcal{F}} = (V,A)$, where $V = N$ and $A = (\cup_{k\in N} A_k)\cup A_0$ where $A_{0} = \{(i,j) \mid i,j \in N \wedge c_{i,j} \neq 0\}$ and $A_k = \{(i,j) \mid i,j \in N \wedge c_{i,j}^k \neq 0\}$.
Each arc $(i,j) \in A_{0}$ is coloured blue and each $(i,j) \in A_{k}$ orange.  
For all $(i,j,R) \in \mathcal{CDS}$ we draw a dotted orange line from node $R$ to arc $(i,j) \in A_R$, denoting that $R$ is the reference bank of the CDS between $i$ and $j$. 
(Strictly speaking, this means that $G_{\mathcal{F}}$ is representable as a coloured directed hypergraph with arcs of size 2 and 3 rather than as a simple graph.)
Finally, we label each arc with the notional of the corresponding contract, and each node with the external assets of the corresponding bank.

\vspace{2mm}
\noindent {\bf Payment schemes.} Banks pay their debts by applying a prespecified rule, called a \emph{payment scheme} which we denote by $\mathbb{P}$. We consider payment schemes that satisfy generalisations \cite{schuldenzucker2016clearing} of two fundamental conditions introduced in~\cite{eisenberg2001systemic} (i.) \emph{Limited Liability}: A bank with sufficient assets to pay all its debts is obliged to do so. 
(ii.) \emph{Absolute Priority}: A defaulting bank submits all of its assets to its creditors.
 
The most studied payment scheme in the literature is the \emph{proportional payment scheme}, where each bank $i$ pays the $r_i$ proportion of each liability, leaving a $(1-r_i)$ fraction of each liability unpaid (studied in, e.g., \cite{ioannidis2022strong, schuldenzucker2016clearing, schuldenzucker2017finding}). 
Recall that the amount a debtor $i$ has to pay to creditor $j$ in a CDS with reference bank $R$ is given by $(1-r_R)c_{i,j}^R$. Consequently, to compute a bank's liabilities, knowledge is required of the recovery rates of other banks (in particular in case $i$ is involved in one or more CDSes as the debtor). 

Given a \emph{recovery rate} vector $r = (r_i)_{i\in N}$, we can express the \emph{liabilities}, \emph{payments}, and \emph{assets} of a financial network where banks pay \emph{proportionally} as follows.  
\begin{description}
\item[Liabilities.] The liability of a bank $i \in N$ to a bank $j \in N$ is $l_{i,j}(r) = c_{i,j} + \sum_{k \in N}(1-r_k)c_{i,j}^k.$ That is we sum up the liabilities from the debt contract and all CDS contracts between $i$ and $j$. We denote by $l_i(r)$ the total liabilities of $i$:
$l_i(r) = \sum_{j \in N}l_{i,j}(r)$. 
\item[Payments.] The payment bank $i$ submits to bank $j$ is denoted by $p_{i,j}(r) = r_i \cdot l_{i,j}(r)$.
\item[Assets.] The assets of a bank $i$ are defined as $a_i(r) = e_i + \sum_{j \in N}p_{j,i}(r)$, namely the sum of its external assets and incoming payments 
\end{description}
We are interested in specific recovery rate vectors called \emph{clearing}. 
\begin{definition}[Clearing recovery rate vector~(\emph{CRRV})] Given a financial network $(N,e,c)$, a recovery rate vector $r \in [0,1]^n$ is called \emph{clearing} iff for all banks $i \in N$, 
\begin{equation}\label{eq:clearing}
    r_i = \min \left\{1,\frac{a_i(r)}{l_i(r)}\right\}, \text{if $l_i(r) > 0$, and $r_i = 1$ if $l_i(r) = 0$.} 
\end{equation}
\end{definition}
For compatibility with  \cite{schuldenzucker2017finding,ioannidis2022strong} we consider only \emph{non-degenerate} networks.\footnote{The hardness result presented in this paper remains unaffected by this condition.}

\begin{definition}[Non-degeneracy \cite{schuldenzucker2017finding,ioannidis2022strong}]\label{def:nondegenerate}
A financial network is \emph{non-degenerate} iff every reference bank is the debtor of at least one debt contract and every CDS debtor either has positive external assets or is debtor in at least one debt contract.
\end{definition}

The computational problem of computing a clearing vector is defined in \cite{schuldenzucker2017finding,ioannidis2022strong}.

\begin{definition}[\proportional~\cite{ioannidis2022strong}]
Given a non-degenerate financial network $\mathcal{F} = (N,e,c)$, compute a clearing recovery rate vector ({CRRV}) whenever the banks pay proportionally. 
\end{definition}

\proportional\ is a total search problem~\cite{schuldenzucker2017finding} where the solutions of an instance $I$ correspond to the fixed points of the following function $(f_I)_i(r)=a_i(r)/\max\{a_i(r), l_i(r)\},\forall i\in N$~\cite{ioannidis2022strong}

Due to the \emph{non-degeneracy} condition the above function is well-defined.\footnote{The function $f_I(r)_i$ is well-defined for nodes $i$ that are not sinks with 0 external assets. Sink nodes $j$
have recovery rate $r_j=1$, cf. \eqref{eq:clearing}. Hence, we can turn these variables $r_j$'s to constants in $f_I$, thus bypassing $0/0$ cases and preserving the continuity of $f_I$.} It is proved that \proportional\ is \textsf{FIXP}-complete (strong apx $\textsf{FIXP}_a$ -complete)~\cite{ioannidis2022strong}. For the  \emph{weak} version it is established that there exist a small unspecified constant under which the problem is \textsf{PPAD}-complete \cite{schuldenzucker2017finding}. 

In this paper we study the \emph{weak} approximation version of the \proportional\ . Given 
$I\in \proportional$ we define the \emph{weak} $\epsilon$-approximate clearing recovery rate vector  ($\epsilon$-\emph{CRRV}) of $I$ as follows: %


\begin{definition}[Weak $\epsilon$-approximate clearing recovery rate vector~($\epsilon$-\emph{CRRV})]\label{epsilon-crrv}
Given a financial network $I = (N,e,c)$, a recovery rate vector $r\in [0,1]^n$ is called weak $\epsilon$-approximate clearing iff for all banks $i\in N$, it holds that (i.) if $e_i > \sum_{j \in N \setminus \{i\}} \left( c_{i,j} + \sum_{k \in N\setminus\{i,j\}} c_{i,j}^k\right)$, then $r_i = 1$; (ii.) otherwise, $ \|r_i-f_I(r_i)\|_{\infty} \leq \epsilon$,
 where $f_I$ is the function induced by \eqref{eq:clearing} with respect to instance $I$ defined above.
 \end{definition}
 Condition \emph{(i.)} guarantees that a bank $i$ with a ``trivial'' recovery rate value of $1$ under any clearing vector, is indeed set to $1$ under this notion of weak $\epsilon$-approximate clearing. For those banks, $r_i$ can be omitted from $f_I$ as a variable, so that our notion of an $\epsilon$-\emph{CRRV} can strictly correspond to a weak $\epsilon$-approximate fixed point of the function $f_I$ in the standard sense.

Next we define $\epsilon$-\proportional\ the approximate version of \proportional\ .

\begin{definition}[$\epsilon$-\proportional]
Given a financial network $\mathcal{F} \in \proportional$, compute a weak $\epsilon$-approximate clearing recovery rate vector ($\epsilon$-CRRV) whenever banks pay proportionally.
\end{definition}
\end{section} 

\section{Inapproximability of clearing vectors}
We present the \pure\ problem from which we reduce to $\epsilon$-$\proportional$. 
The main result establishes that $\epsilon$-\proportional\ is \textsf{PPAD}-hard for $\epsilon \leq \frac{3 - \sqrt 5}{16} \approx 0.048$. The polynomial-time reduction is based on simulating the function of gates in a \pure\ instance with weakly approximate clearing vectors of specific financial subnetworks. Towards this, we provide a mapping of the $[0,1]$ space to the values $\{0,1,\perp \}$, which depends on a parameter $\delta\in (0,1/2)$. We proceed by simulating each gate $g$ with a financial network denoted as $\mathcal{F}_g$, in the sense that an $\epsilon$-\emph{CRRV} of $\mathcal{F}_g$ can be mapped back to an assignment x of values for variables in $g$ using the aforementioned mapping. Finally we choose the mapping that maximises the inapproximability level. We generalise our result to a wider set of payment schemes.
\subsection{ The PURE-CIRCUIT problem }
 An instance $I = (V,G)$ of \pure\ consists of a set of gates $G$ and variables $V$. The gates are represented by tuples $g = (T,u,w)$ or $g = (T,u,v,w)$, where $T$ represents the gate type and $u,v,w \in V$ are either input or output variables, depending on $T$. Our analysis is based on the following gates:
 \begin{itemize}[topsep=0pt, noitemsep]

 \item NOT-gate: $(\text{NOT},u,w)$, where  $u$ is the input variable and $w$ is the output variable.

 \item OR-gate: $(\text{OR},u,v,w)$, where $u$ and $v$ are the input variables and $w$ is the output variable.
 \item PURIFY-gate: $(\text{PURIFY},u,v,w)$, where $u$ is the input and $v,w$ are the output variables.
 \end{itemize}
 
A graphical illustration is presented in Appendix \ref{purecircuitinstance}. Each variable appearing in a gate $g$ is assigned a value from the set $\{0,1,\perp\}$, where the value $\perp$ is termed ``garbage''. We denote an assignment of values to variables by the function $\text{x}: V \mapsto \{0,1,\perp\}$. We are interested in assignments that \emph{satisfy} the above gates.
 \begin{itemize}[topsep=0pt, noitemsep]

 \item An assignment x satisfies gate $(\text{NOT},u,w)$ iff:
 \begin{enumerate}[topsep=0pt, noitemsep]
 \item $\text{x}[u] = 0 \rightarrow \text{x}[w] = 1$
 \item $\text{x}[u] = 1\rightarrow \text{x}[w] =0$
 \item $\text{x}[u] =\ \perp\ \rightarrow \text{x}[w] \in \{0,1,\perp\}$
 \end{enumerate}

 \item  An assignment x satisfies gate $(\text{OR},u,v,w)$ iff: 
 \begin{enumerate}[topsep=0pt, noitemsep]
 \item $\text{x}[u] = \text{x}[v] = 0 \rightarrow \text{x}[w] = 0$.
 \item $\text{x}[u] = 1 \text{ and } \text{x}[v] \in \{0,1,\perp\} \rightarrow \text{x}[w] = 1$
  \item $\text{x}[u] \in \{0,1,\perp\}  \text{ and } \text{x}[v] = 1  \rightarrow \text{x}[w] = 1$
  \item Else $\text{x}[w] = \{0,1,\perp\}$
 \end{enumerate}

 \item   An assignment x satisfies gate $(\text{PURIFY},u,v,w)$ iff:
 \begin{enumerate}[topsep=0pt, noitemsep]
 \item $\text{x}[u] = 0 \rightarrow \text{x}[v] = \text{x}[w] = 0$.
 \item  $\text{x}[u] = 1 \rightarrow \text{x}[v] = \text{x}[w] = 1$.
 \item $\text{x}[u] =\ \perp\ \rightarrow \text{x}[v]  \in \{0,1\} \text{ or } \text{x}[w] \in \{0,1\}$.
 \end{enumerate}
 \end{itemize}
 
Below we define the version of \pure\  and the main theorem we use in this paper.

\begin{definition}[\pure]\label{defpure}
An instance $I = (V,G)$ of \pure\ consists of a set of variables $V$, gates $G$ of the form $g = (T,u,v,w)$ or $g = (T,u,w)$, where $T \in \{\emph{NOT},\emph{OR},\emph{PURIFY}\}$ and no two gates share the same output variable. A solution is an assignment $\emph{x}:V\mapsto \{0,1,\perp\}$ that satisfies all gates in $G$.
\end{definition}

\begin{theorem}[\cite{DBLP:conf/focs/DeligkasFHM22}, Corollary 2.2]
\pure\ is \textsf{PPAD}-complete.
\end{theorem}
\subsection{Reducing \pure\ to \texorpdfstring{$\epsilon$}{Epsilon}-\proportional\  } 
  \smallskip \noindent \textbf{The mapping $m_{\delta}$.} 
  We construct a mapping of $[0,1]$ to  $\{0,1, \perp\}$.
  The domain of this mapping reflects potential recovery rates of banks in a financial network, and these are mapped to values for variables of \pure\ instances. The mapping will be used to translate recovery rate vectors of a financial network to assignments x of a corresponding \pure\ instance. Fix $\delta$ to be any value in $(0,\frac{1}{2})$, and consider the following concrete mapping $m_{\delta}$:
 \begin{itemize}[topsep=0pt, noitemsep]
       
        \item If $r \in 
         \Bigl[0,\frac{1}{2}- \delta\Bigr]$, then $m_\delta(r) = 0$;

        \item If $r \in \Bigl(\frac{1}{2}-\delta, \frac{1}{2} + \delta\Bigr)$, then $m_\delta(r) = \perp$;
       
        \item If $r \in \Bigl[\frac{1}{2} + \delta ,1\Bigr]$, then $m_{\delta}(r) = 1$.
    \end{itemize}
    \vspace{3mm}
We will use $m_{\delta}(r)$ throughout the remainder of the section. A concrete choice of $\delta$ will be defined later, as it will follow from analysis what the optimal choice of $\delta$ is (in the sense of achieving the strongest inapproximability result).

\smallskip \noindent \textbf{A first-order language\footnote{In defining the language we followed the notation in \cite{DBLP:series/mcs/Fitting90}.}.}%
Let $L(\textbf{R},\textbf{F},\textbf{C})$ be a \emph{first-order language} where $\textbf{R} = \{\preceq\}$ is a \emph{relation symbol set}, $\textbf{F} = \{\pm,+,-,\cdot,/\}$ is a \emph{function symbol set} and $\textbf{C} = \emptyset$ a \emph{constant symbol set}. The \emph{model} $M =  \langle \textbf{D},\textbf{I}\rangle$ consists of the \emph{domain} $\textbf{D} = \{[x,y]\lvert x, y\in\mathbb{R}\text{ and }x\leq y\}\cup\{\emptyset\}$. We denote ``\textbf{x}'' variables that represent intervals and ``$x$'', variables that represent numbers\footnote{Strictly speaking  $\textbf{x}^\textbf{A}$ is the interval represented by the symbol $\textbf{x}$ under an assignment \textbf{A} of the symbol to the domain \textbf{D}. We rather use $\textbf{x}$ to also refer to the interval itself in order to prevent notational overflow.}. The \emph{interpretation} $\textbf{I}$ associates symbol $\preceq$ to the relation: 
\begin{equation}\label{relation}
\preceq^{\textbf{I}} = \Bigl\{(\textbf{x},\textbf{y})\mid \inf\{\textbf{x}\}\leq \inf\{\textbf{y}\} \text{ and } \sup\{\textbf{x}\}\leq \sup\{\textbf{y}\}\Bigr\} 
\end{equation}
Function symbols operating on two variables that represent numbers are interpreted to the corresponding mathematical operation within the real number field. 
\begin{itemize}

\item If $x,\epsilon$ represent real numbers with $\epsilon>0$ then 
\begin{equation}\label{pm operation on numbers}
x\pm^{\textbf{I}} \epsilon =
\begin{cases}

  \Bigl[x-\epsilon,x+\epsilon\Bigr]\cap \Bigl[0,1\Bigr], & \text{if } x\in\Bigl[0,1\Bigr]\\

\Bigl[1-\epsilon,1\Bigr]\cap\Bigl[0,1\Bigr], & \text{if } x > 1\\

  \vspace{2mm}
  \Bigl[0,\epsilon\Bigr]\cap\Bigl[0,1\Bigr],& \text{ if } x < 0\\
\end{cases}
\end{equation}
\item If variable $\textbf{x}$ represents an interval within $[0,1]$ and $\epsilon > 0$ represents a real number then

\begin{equation}\label{pm operation on interval}
\textbf{x}\pm^{\textbf{I}} \epsilon = \Bigl[\inf\{\textbf{x}\}-\epsilon,\sup\{\textbf{x}\}+\epsilon\Bigr]\cap\Bigl[0,1\Bigr]
\end{equation}

\item If variable $\textbf{x}$ represents an interval within $[0,1]$ then 
\begin{equation}\label{interval substraction}
1-^\textbf{I}\textbf{x} = \Bigl[1-\sup\{\textbf{x}\},1-\inf\{\textbf{x}\} \Bigr] 
\end{equation}
\item If $\textbf{x} = x\pm\epsilon_1$ and $\textbf{y} = y\pm\epsilon_2$ represent intervals  within the range $[0,1]$ then
\begin{equation}\label{interval addition}
\textbf{x}+^\textbf{I}\textbf{y} =
\begin{cases}
\vspace{1mm}
     \Bigl[ \inf\{\textbf{x}\}+\inf\{\textbf{y}\},\sup\{\textbf{x}\}+\sup\{\textbf{y}\}\Bigr]\cap\Bigl[0,1\Bigr],&\text{ if } x+y\in[0,1]\\
     \vspace{1mm}
     \Bigl[1-(\epsilon_1+\epsilon_2),1)\Bigr]\cap\Bigl[0,1\Bigr],&\text{ if } x+y>1\\
     \vspace{1mm}
      \Bigl[0,\epsilon_1+\epsilon_2\Bigr]\cap\Bigl[0,1\Bigr],&\text{ if } x+y<0
     \end{cases}
\end{equation}

\item If variable $l\geq 1$ and $\textbf{x} = x\pm\epsilon$ represents an interval within $[0,1]$, then 
\begin{equation}\label{number interval mult}
l\cdot^{\textbf{I}}\textbf{x} =
\begin{cases}

\vspace{1mm}
 \Bigr[l\cdot \inf\{\textbf{x}\}, l\cdot \sup\{\textbf{x}\}\Bigr]\cap\Bigl[0,1\Bigr], & \text{ if }l\cdot x \in [0,1] \\
 
 \vspace{1mm}
 \Bigl[1-l\epsilon,1 \Bigr]\cap\Bigl[0,1\Bigr],& \text{ if } l\cdot x > 1\\

 \vspace{1mm}
 \Bigl[0,l\epsilon \Bigr]\cap\Bigl[0,1\Bigr],& \text{ if } l\cdot x <0
 \end{cases}
\end{equation}
\end{itemize}

\begin{remark} Operation $\pm^\textbf{I}$ creates a closed $\epsilon$-ball around $x$, restricted to values that fall within $[0,1]$, with the quirk that the ball is centered to 1 if $x>1$ and 0 if $x<0$. The result of $\pm^{\textbf{I}}$ on two numbers is always a range within $[0,1]$. All functions map intervals within $[0,1]$ to intervals within $[0,1]$.\footnote{The other cases are mapped to $\emptyset$.} 
\end{remark}

A \emph{valid substitution} $\sigma$, 
is a mapping of a term $\tau$ to a term $\sigma(\tau)$, that satisfies the property $(\sigma(\tau))^\textbf{I,A} = \tau^\textbf{I,A}$, where \textbf{A} is an assignment of variables in $\tau$ to values in \textbf{D}\footnote{We allow a certain degree of informality, without compromising accuracy, in using assignments \textbf{A}. A precise description of them falls outside the scope of the paper.}. In essence, a substitution between two terms signifies that, under the interpretation \textbf{I} and any assignment \textbf{A}, both terms refer to the same value within the domain \textbf{D}. We use the following lemmas in our proof. (Appendix \ref{apx:C}).

\begin{lemma}\label{lemma1}The following are valid substitutions.
 \begin{enumerate}
  \label{sigma1-}\item[$\sigma_{1-}$:] If $\textbf{x} = x\pm\epsilon$ then $\sigma_{1-}(1-\textbf{x}) = (1-x) \pm \epsilon$.
  \label{sigmapm}\item[$\sigma_\pm$:]If $\textbf{x} = x\pm\epsilon_1$ and $\epsilon_2>0$, then $\sigma_{\pm}(\textbf{x}\pm \epsilon_2) =  x\pm(\epsilon_1+\epsilon_2)$.
  \label{sigma+}\item[$\sigma_{+}$:] If $\textbf{x} = (x\pm \epsilon_1)$ and $\textbf{y} = (y \pm \epsilon_2)$, then $\sigma_{+}(\textbf{x} + \textbf{y}) = (x+y)\pm (\epsilon_1+\epsilon_2
   )$. 
  \label{sigma*} \item[$\sigma_{*}$:] If  $\textbf{x} = x\pm\epsilon$ and $l\geq1$, then   $\sigma_{*}(l\cdot \textbf{x}) = l\cdot x \pm l\cdot \epsilon$.
   \end{enumerate}
\end{lemma}
 
\begin{lemma}\label{lemma2}
Assume variables $x,y$ represent numbers in \textbf{D} with $x\geq y$ and let $\textbf{x} = x\pm\epsilon$ and $\textbf{y} = y\pm\epsilon$, for $\epsilon>0$. For all assignments \textbf{A} that satisfy these assumptions we say $\textbf{y}\preceq \textbf{x}$ is valid  if $( \textbf{y}^{\textbf{I,A}}, \textbf{x}^{\textbf{I,A}} )\in \preceq^\textbf{I}$.
\end{lemma}

\smallskip \noindent \textbf{Gate simulation.} For each gate $g \in \{\text{NOT},\text{OR},\text{PURIFY}\}$ we construct a financial network $\mathcal{F}_g$ that contains a set of banks that include a bank $b_u$ for each input or output variable $u$ of $g$. We refer to a bank $b_u$ as an \emph{input bank} iff $u$ is an input variable of $g$, and we refer to bank $b_u$ as an \emph{output bank} iff $u$ is an output variable of $g$. Two notable characteristics of all $\mathcal{F}_g$ networks are that all CDS debtor banks possess enough external assets so that their recovery rate is equal to 1 in any $\epsilon$-\emph{CRRV} and no bank has more than one outgoing arc.

We proceed by describing each of the financial networks and we analyse the recovery rate of the output banks under any $\epsilon$-\emph{CRRV}, given a value of the input bank. We denote $r_i$ the recovery rate of bank $i$ and \textbf{r}$_i$ the range of the recovery rate values under $\epsilon$-\emph{CRRV}.  

\begin{remark}
It is significant to realise that the error parameter $\epsilon > 0 $, defines a range of values within which the recovery rate values for banks fall under an $\epsilon$-CRRV. Definition \ref{epsilon-crrv} combined with the fact that all recovery rate values are restricted within $[0,1]$, implies that obtaining an $\epsilon$-CRRV for $\epsilon>\frac{1}{2}$ is achieved by setting $r_i = \frac{1}{2}$ for all banks $i$. Consequently we consider parameter values $\epsilon$ that fall within $(0,\frac{1}{2})$. 
\end{remark}

\smallskip \noindent \fbox{\textbf{NOT-gate.}}
We simulate the gate $(\text{NOT},u,w)$ with the financial network $\mathcal{F}_{\text{NOT}}$, defined in Figure \ref{fig:NOTGATE}. Here, the input bank is $b_u$ and the output bank is $b_w$. 

   \begin{figure}[htbp!]
    \centering
    \scalebox{1}{
\begin{tikzpicture}
[shorten >=1pt,node distance=2cm,initial text=]
\tikzstyle{every state}=[draw=black!50,very thick]
\tikzset{every state/.style={minimum size=0pt}}
\tikzstyle{accepting}=[accepting by arrow]

\node(1) {$b_u$};
\node(2)[right of=1]{1};
\draw[blue,->,very thick] (1)--node[midway,black,yshift=2mm]{1}(2);

\node(3)[below of=1,yshift=5mm]{$2$};
\node[teal,left of=3,xshift=1.5cm]{$\frac{2}{1+2\delta}$};
\node(4)[right of=3]{3};
\node(5)[right of=4]{4};
\draw[blue,->,very thick] (4)--node[midway,black,yshift=2mm]{1}(5);
\draw[orange,->,very thick,snake=snake] (3)--node[midway,black,yshift=4mm,xshift=3mm]{$\frac{2}{1+2\delta}$}(4);
\path[orange,->,draw,dashed,thick] (1) -- ($ (3) !.5! (4) $);

\node(6)[below of =3]{5};
\node(7)[right of=6]{6};
\draw[orange,->,very thick,snake=snake](6)--node[midway,black,yshift=4mm,xshift=-3mm]{$\frac{1+2\delta}{4\delta}$}(7);
\node[teal,left of=6,xshift=1.5cm]{$\frac{1+2\delta}{4\delta}$};
\node(8)[right of=7]{7};
\draw[blue,->,very thick] (7)--node[midway,black,yshift=2mm]{1}(8);
\path [orange,->,draw,dashed,thick] (4) -- ($ (6) !.5! (7) $);

\node(9)[below of =6]{8};
\node[teal,left of=9,xshift=1.7cm]{1};
\node(10)[right of=9]{$b_w$};
\draw[orange,->,very thick,snake=snake] (9)--node[midway,black,yshift=3mm]{1}(10);
\node(11)[right of =10]{9};
\draw[blue,->,very thick] (10)--node[midway,black,yshift=2mm]{1}(11);
\path [orange,->,draw,dashed,thick] (7) -- ($ (9) !.5! (10) $);
\end{tikzpicture}}
\caption{The financial network $\mathcal{F}_{\text{NOT}}$ that simulates a NOT-gate.}
    \label{fig:NOTGATE}
    \end{figure}
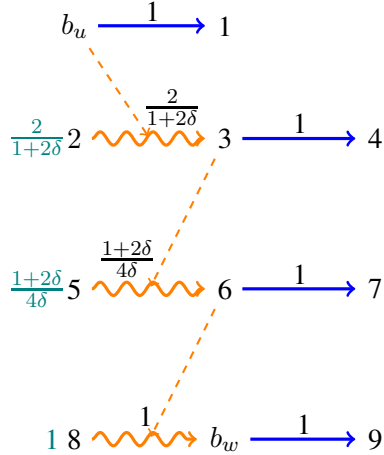
    
Assume that $r = (r_i)_{i\in N}$ is an $\epsilon$-\emph{CRRV} of $\mathcal{F}_{\text{NOT}}$. We will generate the atomic formulas of $L(\text{R},\text{F},\text{C})$ that, when interpreted through \textbf{I}, defines the range of recovery rate values for the banks of $\mathcal{F}_{\text{NOT}}$ under $r=(r_i)_{i\in N}$. Bank 3 receives a payment of $p_{2,3}(r) = \frac{2}{1+2\delta}\cdot[1-r_{b_u}]$ from bank 2 and holds a liability $l_{3,4}(r) = 1$ towards bank 4. The recovery rate values for bank 3 lie within the following range:
\begin{equation}
\text{range}(r_3) =  \Bigl[\min\Bigl(1,\frac{2}{1+2\delta}\cdot[1-r_{b_u}]\Bigr)-\epsilon,\min\Bigl(1,\frac{2}{1+2\delta}\cdot[1-r_{b_u}]\Bigr)+\epsilon\Bigr]\cap\Bigl[0,1\Bigr]
\end{equation}
Applying the proposed notation, the set of values mentioned above can be succinctly represented by the atomic formula $\textbf{r}_3 = \frac{2}{1+2\delta}\cdot[1-r_{b_u}]\pm \epsilon$. Notice that, as $r = (r_i)_{i\in N}$ is presumed to be an $\epsilon$-\emph{CRRV} of $\mathcal{F}_{\text{NOT}}$, the payment received by bank 6 from bank 5 is contingent on the range of $r_3$. Specifically, each value $\rho$ within $\text{range}(r_3)$, establishes a corresponding payment, $p_{5,6}(r) =\frac{1+2\delta}{4\delta}\cdot[1-\rho]$, from bank 5 to bank 6. Consequently, all possible recovery rate values for bank 6 must lie within the range:
\begin{equation}
 \text{range}(r_6) =\bigcup_{\rho\in \text{range}(r_3)} \Bigl[\min\Bigl(1,\frac{1+2\delta}{4\delta}\cdot [1-\rho]\Bigr)-\epsilon, \min\Bigl(1,\frac{1+2\delta}{4\delta}\cdot [1-\rho]\Bigr)+\epsilon\Bigr]\cap\Bigl[0,1\Bigr]
 \end{equation}
 which is represented by the atomic formula $\textbf{r}_{6}$ = $\frac{1+2\delta}{4\delta}\cdot[1-\textbf{r}_3]\pm\epsilon$. Similarly, bank $b_w$ under the $\epsilon$-\emph{CRRV}, receives a payment of $p_{8,b_{w}}(r) = 1-\rho$, where $\rho$ is a value in range$(r_6)$, thus we get
 \begin{equation}
 \text{range}(r_{b_w}) = \bigcup_{\rho\in \text{range}(r_6)}\Bigl[\min\Bigl(1, 1-\rho\Bigr)-\epsilon, \min\Bigl(1, 1-\rho\Bigr)+\epsilon\Bigr]\cap\Bigl[0,1\Bigr]
 \end{equation}
 which is generated by the atomic formula $\textbf{r}_{b_w} = [1-\textbf{r}_6] \pm \epsilon$. To establish the connection between the range of recovery rate values of $b_w$ w.r.t $r_{b_u}$, we apply a series of \emph{valid substitutions} starting  from the formula $\textbf{r}_{b_w} = [1-\textbf{r}_6] \pm \epsilon$. We presen the sequence of substitutions  below.
\begin{equation*}
\begin{split}
 \textbf{r}_{b_w} = [1-\textbf{r}_6] \pm \epsilon & =  \Bigl[1-\Bigl(\frac{1+2\delta}{4\delta}\cdot [1-\textbf{r}_3]\pm\epsilon\Bigr)\Bigr]\pm\epsilon 
  \\ & = \Bigl[1-\Bigl(\frac{1+2\delta}{4\delta}\cdot \underbrace{\Bigl[1-\frac{2}{1+2\delta}\cdot(1-r_{b_u})\pm\epsilon}_{\textbf{r}_3}\Bigr]\pm\epsilon\Bigr)\Bigr]\pm\epsilon\\
 & \xrightarrow{\sigma_{*}(\frac{1+2\delta}{4\delta}\cdot \textbf{r}_3)}\Bigl[1-\Bigl(\underbrace{\Bigl[\frac{1}{2\delta}\cdot r_{b_u}+\frac{-1+2\delta}{4\delta}\pm\frac{1+2\delta}{4\delta}\epsilon\Bigr]}_{\textbf{x}}\pm\epsilon\Bigr)\Bigr]\pm\epsilon\\
 & \xrightarrow{\sigma_{\pm}(\textbf{x}\pm\epsilon)}\Bigl[1-\Bigl(\underbrace{\frac{1}{2\delta}\cdot r_{b_u}+\frac{-1+2\delta}{4\delta}\pm\frac{1+6\delta}{4\delta}\epsilon}_{\textbf{y}}\Bigr)\Bigr]\pm\epsilon\\
 & \xrightarrow{\sigma_{1-}(1-\textbf{y})}\Bigl[\underbrace{-\frac{1}{2\delta}\cdot r_{b_u}+\frac{1+2\delta}{4\delta}\pm\frac{1+6\delta}{4\delta}\epsilon}_{\textbf{z}}\Bigr]\pm\epsilon\\
 & \xrightarrow{\sigma_{\pm}(\textbf{z}\pm\epsilon)}-\frac{1}{2\delta}\cdot r_{b_u}+\frac{1+2\delta}{4\delta}\pm\frac{1+10\delta}{4\delta}\epsilon
\end{split}
\end{equation*}
Eventually the range for the recovery rate values of the output bank $b_w$ in any $\epsilon$-\emph{CRRV} is generated from an assignment \textbf{A} where $r_{b_u}\in[0,1]\text{ and }\delta,\epsilon\in (0,\frac{1}{2})$ and by interpreting the following formula w.r.t. \textbf{I}: 
 \begin{equation}\label{eq:NOT}
    \textbf{r}_{b_w} = -\frac{1}{2\delta}\cdot r_{b_u}+ \frac{1+2\delta}{4\delta} \pm \frac{1+10\delta}{4\delta}\epsilon.
\end{equation}

We proceed by computing the interval within which $r_{b_w}$ belongs in any $\epsilon$-\emph{CRRV}, depending on the interval where the recovery rate of the input bank $r_{b_u}$ belongs according to the mapping $m_\delta$. In all assignments \textbf{A} assumed in the proof of the claims, both $\delta$ and $\epsilon$ belong in $(0,\frac{1}{2})$.  
    \begin{Claim}\label{NOT-CLAIM}
    Consider the financial network $\mathcal{F}_{\text{NOT}}$ of Figure \ref{fig:NOTGATE} and let $r = (r_i)_{i\in N}$ be an $\epsilon$-CRRV:
    \begin{enumerate}[topsep=0pt, noitemsep]
    \item If $r_{b_u} \in \Bigl[0,\frac{1}{2}-\delta\Bigr]$ then $r_{b_w} \in \Bigl[1-\frac{1+10\delta}{4\delta}\epsilon,1\Bigr]$.

    \vspace{1mm}
    \item If $r_{b_u} \in \Bigl[\frac{1}{2}+\delta,1\Bigr]$ then $r_{b_w}\in \Bigl[0,\frac{1+10\delta}{4\delta}\epsilon\Bigr]$.
    \end{enumerate}
     \end{Claim}
      \begin{proof}
    We use (\ref{eq:NOT}), that expresses range($r_{b_w}$) in any $\epsilon$-\emph{CRRV} in terms of $r_{b_u}$. 
    \begin{enumerate}[topsep=0pt, noitemsep]
\item In any assignment \textbf{A} of variables in \ref{eq:NOT} to values  in \textbf{D} where $r_{b_u} \leq \frac{1}{2}-\delta$, from Lemma \ref{lemma2} and after calculations, it holds that $ \Bigl(1\pm \frac{1+10\delta}{4\delta}\epsilon\Bigr)^\textbf{I,A} \preceq^\textbf{I} \Bigl(-\frac{1}{2\delta}\cdot r_{b_u}+ \frac{1+2\delta}{4\delta}\pm \frac{1+10\delta}{4\delta}\epsilon\Bigr)^\textbf{I,A}$, namely $1\pm\frac{1+10\delta}{4\delta}\epsilon\preceq\textbf{r}_w$ is \emph{valid}. Interpreting the last relation according to \textbf{I}, from \ref{pm operation on numbers} and \ref{relation} we conclude that under any $\epsilon$-\emph{CRRV}, if $r_{b_u}\in\Bigl[0,\frac{1}{2}-\delta\Bigr]$ then $r_{b_w}\in \Bigl[1-\frac{1+10\delta}{4\delta}\epsilon,1\Bigr]$. 

\vspace{1mm}
\item In any assignment \textbf{A} of variables in \ref{eq:NOT} to values  in \textbf{D} where $r_{b_u}\geq \frac{1}{2}+\delta$, from Lemma \ref{lemma2} and after calculations, it holds that $\Bigl(- \frac{1}{2\delta}\cdot r_{b_u} + \frac{1+2\delta}{4\delta} \pm \frac{1+10\delta}{4\delta}\epsilon\Bigr)^\textbf{I,A}\preceq^\textbf{I} \Bigl(0\pm\frac{1+10\delta}{4\delta}\epsilon\Bigr)^{\textbf{I,A}}$, namely $\textbf{r}_{b_w}\preceq 0\pm\frac{1+10\delta}{4\delta}\epsilon$ is \emph{valid}. Applying \textbf{I} and from  \ref{pm operation on numbers} and \ref{relation}, under any $\epsilon$-\emph{CRRV} if $r_{b_u}\in\Bigl[\frac{1}{2}+\delta,1\Bigr]$ then $r_{b_w}\in \Bigl[0,\frac{1+10\delta}{4\delta}\epsilon\Bigr]$
\qedhere
\end{enumerate}
   \end{proof}

\noindent \fbox{\textbf{OR-gate.}} We simulate the gate $(\text{OR},u,v,w)$ with the financial network $\mathcal{F}_{\text{OR}}$ of Figure \ref{fig:ORGATE}. The network consists of two input banks $b_u$ and $b_v$ corresponding to the input variables $u,v$ respectively, and an output bank $b_w$ corresponding to the output variable $w$. 

\begin{figure}[htbp!]
    \centering
    \scalebox{1}{
\begin{tikzpicture}
[shorten >=1pt,node distance=2cm,initial text=]
\tikzstyle{every state}=[draw=black!50,very thick]
\tikzset{every state/.style={minimum size=0pt}}
\tikzstyle{accepting}=[accepting by arrow]

\node(1) {$b_u$};
\node(2)[right of=1]{1};
\draw[blue,->,very thick] (1)--node[midway,black,yshift=2mm]{1}(2);
\node(3)[below of=1,yshift=5mm]{$2$};
\node[teal,left of=3,xshift=1.5cm]{$\frac{2}{1+2\delta}$};
\node(4)[right of=3]{$3$};
\node(5)[right of=4]{4};
\draw[blue,->,very thick] (4)--node[midway,black,yshift=2mm]{1}(5);
\draw[orange,->,very thick,snake=snake] (3)--node[midway,black,yshift=4mm,xshift=3mm]{$\frac{2}{1+2\delta}$}(4);
\path [orange,->,draw,dashed,thick] (1) -- ($ (3) !.5! (4) $);

\node(6)[below of =3]{5};
\node(7)[right of=6]{6};
\draw[orange,->,very thick,snake=snake] (6)--node[midway,black,yshift=4mm,xshift=-3mm]{$\frac{1+2\delta}{4\delta}$}(7);
\node[teal,left of=6,xshift=1.5cm]{$\frac{1+2\delta}{4\delta}$};
\path [orange,->,draw,dashed,thick] (4) -- ($ (6) !.5! (7) $);

\node[right of= 2,xshift=1.5cm](11) {$b_v$};
\node(12)[right of=11]{7};
\draw[blue,->,very thick] (11)--node[midway,black,yshift=2mm]{1}(12);

\node(13)[below of=11,yshift=5mm]{8};
\node[teal,left of=13,xshift=1.5cm]{$\frac{2}{1+2\delta}$};
\node(14)[right of=13]{9};
\node(15)[right of=14]{10};
\draw[blue,->,very thick] (14)--node[midway,black,yshift=2mm]{1}(15);
\draw[orange,->,very thick,snake=snake] (13)--node[midway,black,yshift=4mm,xshift=3mm]{$\frac{2}{1+2\delta}$}(14);
\path [orange,->,draw,dashed,thick] (11) -- ($ (13) !.5! (14) $);

\node(16)[below of =13]{12};
\node(17)[right of=16]{11};
\draw[orange,->,very thick,snake=snake] (17)--node[midway,black,yshift=4mm,xshift=-3mm]{$\frac{1+2\delta}{4\delta}$}(16);
\node[teal,right of=17,xshift=-1.5cm]{$\frac{1+2\delta}{4\delta}$};
\path [orange,->,draw,dashed,thick] (14) -- ($ (16) !.5! (17) $);

\node(20)[left of=16,xshift=2mm]{$b_w$};
\draw[blue,->,very thick](7)--node[midway,black,yshift=2mm]{1}(20);
\draw[blue,->,very thick](16)--node[midway,black,yshift=2mm]{1}(20);

\node(21)[above of=20,yshift=-7mm]{$13$};
\draw[blue,->,very thick](20)--node[midway,black,xshift=-2mm]{1}(21);

\end{tikzpicture}}
\caption{The financial network $\mathcal{F}_{\text{OR}}$ that simulates an OR-gate}
    \label{fig:ORGATE}
\end{figure}
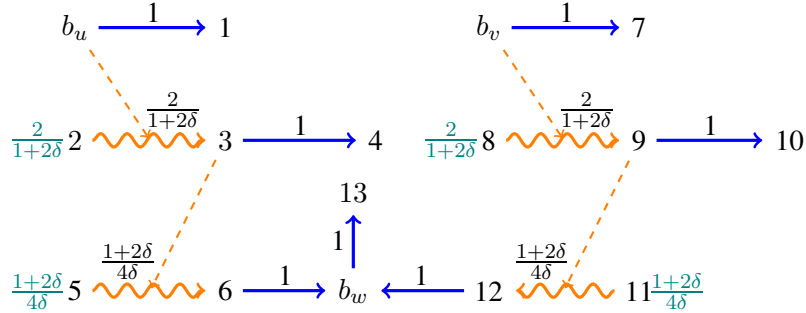

Assume that $r = (r_i)_{i\in N}$ is an $\epsilon$-\emph{CRRV} of $\mathcal{F}_{\text{OR}}$ and let the recovery rates of the input banks $b_{u},b_{v}$ to be $r_{b_u}\text{ and }r_{b_v}$ respectively. The range of recovery rate values for banks 3 and 9 in $\mathcal{F}_{\text{OR}}$ matches that of bank 3 within the network $\mathcal{F}_{\text{NOT}}$. Symmetrically range($r_6$) and range($r_{12}$) in $\mathcal{F}_{\text{OR}}$ matches range($r_6$) within the network $\mathcal{F}_{\text{NOT}}$. Both 6 and 12 submit to $b_w$ payments of at most $\rho_1,\rho_2$ belonging to $\text{range}(r_6)$ and $\text{range}(r_{12})$ respectively. Consequently the range of recovery rate values of the output bank $b_w$ is:
\begin{equation}
\text{range}(r_{b_w}) = \bigcup_{\substack{\rho_1\in \text{range}(r_6)\\ \rho_2\in{\text{range}}(r_{12})}}\Bigl[\min\Bigl(1,\rho_1+\rho_2\Bigr)-\epsilon, \min\Bigl(1,\rho_1+\rho_2\Bigr)+\epsilon \Bigr]\cap\Bigl[0,1\Bigr] 
\end{equation}
which is generated  from an assignment \textbf{A} where $r_{b_u}\in[0,1]\text{ and }\delta,\epsilon\in (0,\frac{1}{2})$ by applying \textbf{I} to the formula,
\begin{equation}\label{eq:OR}
\textbf{r}_{b_w} = (\textbf{r}_6 + \textbf{r}_{12}) \pm \epsilon
\end{equation} 
where for $\textbf{r}_6$ (and similarly $\textbf{r}_{12}$),
\begin{equation*}
\begin{split}
 \textbf{r}_{6} = \frac{1+2\delta}{4\delta}\cdot[1-\textbf{r}_3] \pm \epsilon
 & = \Bigl(\frac{1+2\delta}{4\delta}\cdot\Bigl[1-\Bigl(\underbrace{\frac{2}{1+2\delta}\cdot[1-r_{b_u}]\pm\epsilon\Bigr)}_{\textbf{r}_3}\Bigr]\Bigr)\pm \epsilon \\
 & \xrightarrow{\sigma_{1-}(1-\textbf{r}_3)}\Bigl(\frac{1+2\delta}{4\delta}\cdot\Bigl[\underbrace{\frac{2}{1+2\delta}\cdot r_{b_u}+\frac{-1+2\delta}{1+2\delta}\pm \epsilon}_{\textbf{x}}\Bigr]\Bigr)\pm\epsilon\\
 & \xrightarrow{\sigma_{*}(\frac{1+2\delta}{4\delta}\cdot\textbf{x})}\Bigl( \underbrace{\frac{1}{2\delta}\cdot r_{b_u}+\frac{-1+2\delta}{4\delta}\pm\frac{1+2\delta}{4\delta}\epsilon}_{\textbf{y}}\Bigr)\pm\epsilon\\
 & \xrightarrow{\sigma_{\pm}(\textbf{y}\pm\epsilon)}\frac{1}{2\delta}\cdot r_{b_u}+\frac{-1+2\delta}{4\delta}\pm\frac{1+6\delta}{4\delta}\epsilon
 \end{split}
\end{equation*}

We proceed by computing the interval where $r_{b_w}$ belongs in any $\epsilon$-\emph{CRRV}, depending on the interval of the mapping $m_\delta$ where the recovery rate of the input bank $r_{b_u}$ belongs. In all assignments \textbf{A}, both $\delta$ and $\epsilon$ belong in $(0,\frac{1}{2})$.     
\begin{Claim}\label{OR-CLAIM}
Consider the financial network $\mathcal{F}_{\text{OR}}$ of Figure \ref{fig:ORGATE} and let $r$ be an $\epsilon$-CRRV.
\begin{enumerate}[topsep=0pt, noitemsep]
\item If $r_{b_u} \in \Bigl[\frac{1}{2}+\delta,1\Bigr]$ or $r_{b_v} \in \Bigl[\frac{1}{2}+\delta,1\Bigr]$ then $r_{b_w} \in \Bigl[1-\frac{1+10\delta}{4\delta}\epsilon,1\Bigr]$.

\vspace{2mm}
\item If $r_{b_u},r_{b_v} \in \Bigl[0,\frac{1}{2}-\delta\Bigr]$ then $r_{b_w} \in \Bigl[0,\frac{1+8\delta}{2\delta}\epsilon\Bigr]$.
\end{enumerate}
\end{Claim}
 \begin{proof}[Proof of Claim 2]
We use (\ref{eq:OR}), that expresses range($r_{b_w}$) in any $\epsilon$-\emph{CRRV} in terms of $r_{b_u},r_{b_v}$.
\begin{enumerate}[topsep=0pt, noitemsep]
\item In any assignment \textbf{A} of variables in \ref{eq:OR} to values in \textbf{D} let w.l.o.g  $r_{b_u} \geq \frac{1}{2}+\delta$. Since according to (\ref{eq:OR}), $\textbf{r}_{b_{w}} = (\textbf{r}_{6} + \textbf{r}_{{12}}) \pm \epsilon$, it is easy to verify that $\Bigl(1 \pm \frac{1+6\delta}{4\delta}\epsilon) \pm \epsilon\Bigr)^\textbf{I,A} \preceq^\textbf{I} \Bigl((\textbf{r}_{6} + \textbf{r}_{12}) \pm \epsilon \Bigr)^\textbf{I,A}$, where by applying substitution $\sigma_{\pm}$ we conclude that $1 \pm \frac{1+10\delta}{4\delta}\epsilon \preceq \textbf{r}_{b_w}$ is \emph{valid}. This implies that under any $\epsilon$-\emph{CRRV}, if $r_{b_u}\in\Bigl[\frac{1}{2}+\delta,1\Bigr]$ then $r_{b_w}\in \Bigl[1-\frac{1+10\delta}{4\delta}\epsilon, 1\Bigr]$.

\vspace{1mm}
\item In any assignment \textbf{A} of variables in \ref{eq:OR} to values in \textbf{D} where both $r_{b_u},r_{b_v} \leq \frac{1}{2}-\delta$, it follows from the above analysis by similar reasoning that both $\textbf{r}_6,\textbf{r}_{12} \preceq 0 \pm \frac{1+6\delta}{4\delta}\epsilon$ are valid. Consequently from (\ref{eq:OR}) we can easily get that $\textbf{r}_{b_w} \preceq (0 \pm 2\frac{1+6\delta}{4\delta}\epsilon)\pm\epsilon$, which by applying $\sigma_{\pm}$ implies that $\textbf{r}_{b_w} \preceq 0 \pm \frac{1+8\delta}{2\delta}\epsilon$ is \emph{valid}. This establishes that under any $\epsilon$-\emph{CRRV}, if both $r_{b_u},r_{b_v}\in\Bigl[\frac{1}{2}+\delta,1\Bigr]$ then $r_{b_w}\in \Bigl[0,\frac{1+8\delta}{2\delta}\epsilon\Bigr]$.\qedhere
\end{enumerate}
\end{proof}

\noindent \fbox{\textbf{PURIFY-gate.}} We simulate the gate $(\text{PURIFY},u,v,w)$ with the financial network of Figure \ref{fig:PURIFYGATE}. The network has one input bank denoted as $b_u$ that corresponds to the input variable $u$ and two output banks denoted as $b_v, b_w$ corresponding to the output variables $v,w$ respectively. 

\begin{figure}[htbp!]
    \centering
    \scalebox{1}{
\begin{tikzpicture}
[shorten >=1pt,node distance=2cm,initial text=]
\tikzstyle{every state}=[draw=black!50,very thick]
\tikzset{every state/.style={minimum size=0pt}}
\tikzstyle{accepting}=[accepting by arrow]
\node(1) {$b_u$};
\node(2)[right of=1]{1};
\draw[blue,->,very thick] (1)--node[midway,black,yshift=2mm]{1}(2);

\node(3)[below right of=1]{5};
\node[teal,left of=3,xshift=1.5cm]{2};
\node(4)[right of=3]{6};
\node(5)[right of=4]{7};
\draw[blue,->,very thick] (4)--node[midway,black,yshift=2mm]{1}(5);
\draw[orange,->,very thick,snake=snake] (3)--node[midway,black,yshift=4mm]{2}(4);

\node(6)[below left of=1]{2};
\node(7)[left of=6]{3};
\draw[orange,->,very thick,snake=snake] (6)--node[midway,black,yshift=4mm,xshift=-3mm]{$\frac{2}{1+2\delta}$}(7);
\node(8)[left of=7]{4};
\draw[blue,->,very thick] (7)--node[midway,black,yshift=2mm]{1}(8);
\node[teal,right of=6,xshift=-1.5cm]{$\frac{2}{1+2\delta}$};
\path [orange,->,draw,dashed,thick] (1) -- ($ (3) !.5! (4) $);
\path [orange,->,draw,dashed,thick] (1) -- ($ (6) !.5! (7) $);

\node(9)[below of=4,xshift=1cm]{9};
\node[teal,right of=9,xshift=-1.5cm]{$\frac{1}{2\delta}$};
\node(10)[left of=9]{$b_w$};
\draw[orange,->,very thick,snake=snake] (9)--node[midway,black,yshift=-4mm]{$\frac{1}{2\delta}$}(10);
\node(11)[below of=7,xshift=-1cm]{8};
\node[teal,left of=11,xshift=1.5cm]{$\frac{1+2\delta}{2\delta}$};
\node(12)[right of=11]{$b_v$};
\draw[orange,->,very thick,snake=snake] (11)--node[midway,black,yshift=-4mm]{$\frac{1+2\delta}{2\delta}$}(12);
\node(13)[right of=12,xshift=4mm]{10};

\draw[blue,->,very thick] (12)--node[midway,black,yshift=2mm]{1}(13);
\draw[blue,->,very thick] (10)--node[midway,black,yshift=2mm]{1}(13);
\path [orange,->,draw,dashed,thick] (7) -- ($ (11) !.5! (12) $);
\path [orange,->,draw,dashed,thick] (4) -- ($ (9) !.5! (10) $);

\end{tikzpicture}}
\caption{The financial network $\mathcal{F}_\text{PURIFY}$ that simulates a PURIFY-gate}
    \label{fig:PURIFYGATE}
\end{figure}
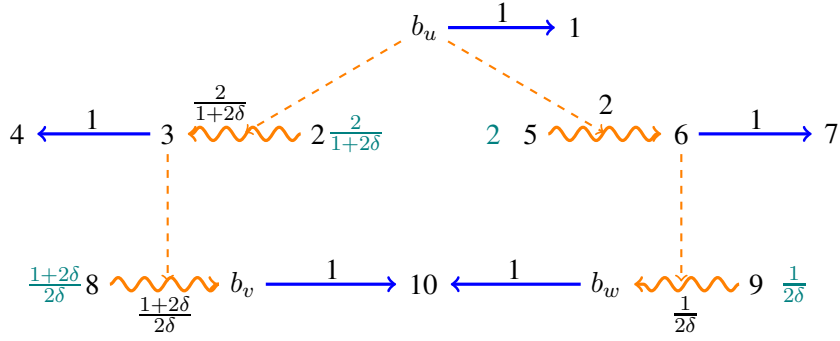

Assume that $r = (r)_{i\in N}$ is an $\epsilon$-\emph{CRRV} of $\mathcal{F}_{\text{PURIFY}}$ and let the recovery rate of the input bank $b_{u}$ be $r_{b_u}$ and the recovery rates of the output banks $b_u, b_w$ be $r_{b_u}, r_{b_w}$ respectively.
\begin{itemize}
\item Let the \emph{Left branch}, be the subnetwork that spans nodes $\{b_u,1,2,3,4,8,b_v,10\}$ . \item Let the \emph{Right branch}, be the subnetwork that spans nodes $\{b_u,1,5,6,7,5,9,b_w,10\}$.
\end{itemize}
We proceed with analysing the intervals for the recovery rates $r_{b_v},r_{b_w}$ in any $\epsilon$-\emph{CRRV} as a function of $r_{b_u}$. Our analysis consists of two parts, one for each branch.

\begin{itemize}[topsep=0pt, noitemsep]
 \item
 
 \underline{\textbf{\emph{Left} branch}} : Bank 3 receives a payment of $\frac{2}{1+2\delta}\cdot [1-r_{b_u}]$ from bank 2 and holds a liability of 1 towards bank 4. The recovery rate values for bank 3 under $r = (r_i)_{i\in N}$ belong in the range:
 \begin{equation}
 \text{range}(r_3) = \Bigl[\min\Bigl(1,\frac{2}{1+2\delta}\cdot[1-r_{b_u}]\Bigr)-\epsilon, \min\Bigl(1,\frac{2}{1+2\delta}\cdot[1-r_{b_u}]\Bigr)+\epsilon \Bigr]\cap\Bigl[0,1\Bigr] 
 \end{equation}
 which is represented by the atomic formula $\textbf{r}_3 = \frac{2}{1+2\delta}\cdot[1-r_3]\pm \epsilon$. Bank $b_{v}$ receives a payment of $\frac{1+2\delta}{2\delta}\cdot[1-\rho]$ from bank 8, where $\rho\in \text{range}(r_3)$. The range of recovery rate values for bank $b_v$ is
 \begin{equation}
 \text{range}(r_{b_v}) =\bigcup_{x\in \text{range}(r_3)} \Bigl[\min\Bigl(1,\frac{1+2\delta}{2\delta}\cdot [1-\rho]\Bigr)-\epsilon, \min\Bigl(1,\frac{1+2\delta}{2\delta}\cdot [1-\rho]\Bigr)+\epsilon\Bigr]\cap\Bigl[0,1\Bigr]
 \end{equation}
 which is represented by the atomic formula $\textbf{r}_{b_v} = \frac{1+2\delta}{2\delta}\cdot[1-\textbf{r}_3]\pm \epsilon$. The atomic formulas and the substitutions that generate the interval for the values of $r_{b_v}$ under the $\epsilon$-\emph{CRRV} are listed below.

\begin{equation*}
\begin{split}
 \textbf{r}_{b_v} = \frac{1+2\delta}{2\delta}\cdot[1-\textbf{r}_3] \pm \epsilon
 & =  \Bigl(\frac{1+2\delta}{2\delta}\cdot\Bigl[1-\Bigl(\underbrace{\frac{2}{1+2\delta}\cdot[1-r_{b_u}]\pm\epsilon\Bigr)}_{\textbf{r}_3}\Bigr]\Bigr)\pm \epsilon \\
 & \xrightarrow{\sigma_{1-}(1-\textbf{r}_3)}\Bigl(\frac{1+2\delta}{2\delta}\cdot\Bigl[\underbrace{\frac{2}{1+2\delta}\cdot r_{b_u}+ \frac{-1+2\delta}{1+2\delta}\pm \epsilon}_{\textbf{x}}\Bigr]\Bigr)\pm\epsilon\\
 & \xrightarrow{\sigma_{*}(\frac{1+2\delta}{2\delta}\cdot\textbf{x})}\Bigl( \underbrace{\frac{1}{\delta}\cdot r_{b_u}+\frac{-1+2\delta}{1+2\delta}\pm\frac{1+2\delta}{2\delta}\cdot\epsilon}_{\textbf{y}}\Bigr)\pm\epsilon\\
 & \xrightarrow{\sigma_{\pm}(\textbf{y}\pm\epsilon)}\frac{1}{\delta}\cdot r_{b_u} + \frac{-1+2\delta}{2\delta} \pm \frac{1+4\delta}{2\delta}\epsilon
 \end{split}
\end{equation*}
 
From the above analysis, the range for the recovery rate values of the output bank $b_v$ in any $\epsilon$-\emph{CRRV} is generated from an assignment \textbf{A} where $r_{b_u}\in[0,1]\text{ and }\delta,\epsilon\in (0,\frac{1}{2})$ and the interpretation of the following formula according to \textbf{I}: %
 \begin{equation}\label{eq:PUR_1}
 \textbf{r}_{b_{v}} = \frac{1}{\delta}\cdot r_{b_u} + \frac{-1+2\delta}{2\delta} \pm \frac{1+4\delta}{2\delta}\epsilon .
 \end{equation}
 
\item \underline{\textbf{\emph{Right} branch}}: Bank 6 receives a payment of $2\cdot[1-r_{b_u}]$ from bank 5 and since it holds a liability of 1 towards bank 7 the range of recovery rate values is given by the following expression:
 \begin{equation}
 \text{range}(r_6) = \Bigl[\min\Bigl(1,2\cdot[1-r_{b_u}]\Bigr)-\epsilon, \min\Bigl(1,2\cdot[1-r_{b_u}]\Bigr)+\epsilon \Bigr]\cap\Bigl[0,1\Bigr] 
 \end{equation}
 which is represented by the atomic formula $\textbf{r}_6 = 2\cdot[1-r_3]\pm \epsilon$. Bank $b_{w}$ receives a payment of $\frac{1}{2\delta}\cdot[1-\rho]$ from bank 9, where $\rho\in \text{range}(r_6)$. Consequently the range of recovery rate values for bank $b_w$ under the $\epsilon$-\emph{CRRV} $r = (r_i)_{i\in N}$ is:
 \begin{equation}
 \text{range}(r_{b_w}) =\bigcup_{\rho\in \text{range}(r_6)} \Bigl[\min\Bigl(1,\frac{1}{2\delta}\cdot [1-\rho]\Bigr)-\epsilon, \min\Bigl(1,\frac{1}{2\delta}\cdot [1-\rho]\Bigr)+\epsilon\Bigr]\cap\Bigl[0,1\Bigr]
 \end{equation}
 which is represented by the atomic formula $\textbf{r}_{b_w} = \frac{1}{2\delta}\cdot[1-\textbf{r}_6]\pm \epsilon$. The atomic formulas and the substitutions that define the range where the values for $r_{b_w}$ belong under the $\epsilon$-\emph{CRRV} are listed below.

 \begin{equation*}
\begin{split}
 \textbf{r}_{b_w} = \frac{1}{2\delta}\cdot[1-\textbf{r}_6] \pm \epsilon
 & = \Bigl(\frac{1}{2\delta}\cdot\Bigl[1-\Bigl(\underbrace{2\cdot[1-r_{b_u}]\pm\epsilon\Bigr)}_{\textbf{r}_6}\Bigr]\Bigr)\pm \epsilon \\
 & \xrightarrow{\sigma_{1-}(1-\textbf{r}_6)}\Bigl(\frac{1}{2\delta}\cdot\Bigl[\underbrace{-1 + 2\cdot r_{b_u}\pm \epsilon}_{\textbf{x}}\Bigr]\Bigr)\pm\epsilon\\
 & \xrightarrow{\sigma_{*}(\frac{1}{2\delta}\cdot\textbf{x})}\Bigl( \underbrace{\frac{1}{\delta}\cdot r_{b_u}-\frac{1}{2\delta}\pm\frac{1}{2\delta}\epsilon}_{\textbf{y}}\Bigr)\pm\epsilon\\
 & \xrightarrow{\sigma_{\pm}(\textbf{y}\pm\epsilon)}\frac{1}{\delta}\cdot r_{b_u} - \frac{1}{2\delta} \pm \frac{1+2\delta}{2\delta}\epsilon
 \end{split}
\end{equation*}

\vspace{2mm}
The range for the recovery rate values of the output bank $b_w$ in any $\epsilon$-\emph{CRRV} follows from an assignment \textbf{A} where $r_{b_u}\in[0,1]\text{ and }\delta,\epsilon\in (0,\frac{1}{2})$ and the interpretation of the following formula according to \textbf{I}: %
\begin{equation}\label{eq:PUR_2}
 \textbf{r}_{b_w} = \frac{1}{\delta}\cdot r_{b_u} - \frac{1}{2\delta} \pm \frac{1+2\delta}{2\delta}\epsilon .
 \end{equation}
\end{itemize}

We compute the range for $r_{b_v}, r_{b_w}$ with respect to $r_{b_u}$ are as follows:
\begin{Claim}\label{PURIFY-CLAIM}
Consider the financial network $\mathcal{F}_{\text{PURIFY}}$ of Figure \ref{fig:PURIFYGATE} and let $r$ be an $\epsilon$-CRRV.
\begin{enumerate}[topsep=0pt, noitemsep]
\vspace{2mm}
\item If $r_{b_u} \in \Bigl[0,\frac{1}{2}-\delta\Bigr]$ then $r_{b_v} \in \Bigl[0,\frac{1+4\delta}{2\delta}\epsilon\Bigr]$ and $r_{b_w} \in \Bigl[0,\frac{1+2\delta}{2\delta}\epsilon\Bigr]$.

\vspace{2mm}
\item If $r_{b_u} \in \Bigl[\frac{1}{2}+\delta,1\Bigr]$ then $r_{b_v}\in \Bigl[1-\frac{1+4\delta}{2\delta}\epsilon,1\Bigr]$ and $r_{b_w} \in \Bigl[1-\frac{1+2\delta}{2\delta}\epsilon,1\Bigr]$.

\vspace{2mm}
\item If $r_{b_u} \in \Bigl(\frac{1}{2}-\delta,\frac{1}{2}+\delta\Bigr)$ then $r_{b_v} \in \Bigl[1-\frac{1+4\delta}{2\delta}\epsilon,1\Bigr]$ or $r_{b_w} \in \Bigl[0,\frac{1+2\delta}{2\delta}\epsilon\Bigr]$.
\end{enumerate}
\end{Claim}
\begin{proof}[Proof of Claim 3]
Similar to the previous claims, the proof follows from a case analysis of (\ref{eq:PUR_1}) and (\ref{eq:PUR_2}) upon values of $r_{b_u}$.
\begin{enumerate}[topsep=0pt, noitemsep]
\item In any assignment \textbf{A} of variables in \ref{eq:PUR_1} to values in \textbf{D} where $r_{b_u} \leq \frac{1}{2}-\delta$,
from Lemma \ref{lemma2} and after calculations we get that in any $\epsilon$-\emph{CRRV}, $\Bigl(\frac{1}{\delta}\cdot r_{b_u} + \frac{-1+2\delta}{2\delta} \pm \frac{1+4\delta}{2\delta}\epsilon\Bigr)^\textbf{I,A}\preceq^\textbf{I} \Bigl(0\pm\frac{1+4\delta}{2\delta}\epsilon\Bigr)^\textbf{I,A}$, namely $\textbf{r}_{b_v} \preceq 0\pm\frac{1+4\delta}{2\delta}\epsilon$ is \emph{valid}. This establishes that if $r_{b_u}\in \Bigl[0,\frac{1}{2}-\delta \Bigr]$ then $r_{b_v}\in \Bigl[0,\frac{1+4\delta}{2\delta}\epsilon\Bigr]$. Similarly combining \ref{eq:PUR_2} with Lemma \ref{lemma2} and the fact that $r_{b_u}\leq \frac{1}{2}$, we conclude that $\textbf{r}_{b_w} \preceq 0 \pm \frac{1+2\delta}{2\delta}\epsilon$ is \emph{valid}. This establishes that whenever $r_{b_u}\in \Bigl[0,\frac{1}{2}-\delta\Bigr]$ then $r_{b_w}\in \Bigl[0,\frac{1+2\delta}{2\delta}\epsilon\Bigr]$.

\vspace{1mm}
\item In any assignment \textbf{A} of variables where $r_{b_u} \geq \frac{1}{2}+\delta$, trivially $r_{b_u} \geq \frac{1}{2}$ thus considering (\ref{eq:PUR_1}) and Lemma \ref{lemma2} and after calculations it holds that $\Bigl(1\pm\frac{1+4\delta}{2\delta}\epsilon\Bigr)^\textbf{I,A} \preceq^\textbf{I}  \Bigl(\frac{1}{\delta}r_{b_u} + \frac{-1+2\delta}{2\delta}\pm \frac{1+4\delta}{2\delta}\epsilon\Bigr)^\textbf{I,A}$, namely $ 1\pm\frac{1+4\delta}{2\delta}\epsilon \preceq \textbf{r}_{b_v}$ is \emph{valid}. This establishes that whenever $r_{b_u}\in \Bigl[\frac{1}{2}+\delta, 1\Bigr]$ then $r_{b_v}\in \Bigl[1-\frac{1+4\delta}{2\delta}\epsilon,1\Bigr]$. Similarly combining the initial assumption with (\ref{eq:PUR_2}) and Lemma \ref{lemma2}, after calculations it holds that $1 \pm \frac{1+2\delta}{2\delta}\epsilon \preceq \textbf{r}_{b_w}$ is \emph{valid}. This establishes that whenever $r_{b_u}\in \Bigl[\frac{1}{2}+\delta,1\Bigr]$, then  $r_{b_w}\in [1-\frac{1+2\delta}{2\delta}\epsilon,1]$

\item Finally for assignments \textbf{A} where $r_{b_u} \in \Bigl(\frac{1}{2}-\delta,\frac{1}{2}+\delta\Bigr)$,  if $r_{b_u} \leq \frac{1}{2}$ then Statement 1 already establishes that $\textbf{r}_{b_w} \preceq 0 \pm \frac{1+2\delta}{2\delta}\epsilon$ is \emph{valid}, while if $r_{b_u}\geq \frac{1}{2}$, then Statement 2 establishes that $1 \pm \frac{1+4\delta}{2\delta}\epsilon \preceq \textbf{r}_{b_v} $ is \emph{valid}. Consequently whenever $r_{b_u}\in \Bigl(\frac{1}{2}-\delta,\frac{1}{2}+\delta\Bigr)$ either $r_{b_u}\in\Bigl[1-\frac{1+4\delta}{2\delta}\epsilon,1\Bigr]$ or $r_{b_w}\in\Bigl[0,\frac{1+2\delta}{2\delta}\epsilon\Bigr]$. \qedhere
\end{enumerate}
\end{proof}

\noindent \textbf{Specifying $\delta$.} Claims \ref{NOT-CLAIM}, \ref{OR-CLAIM} and \ref{PURIFY-CLAIM} establish for the networks $\mathcal{F}_{\text{NOT}}, \mathcal{F}_{\text{OR}}, \mathcal{F}_{\text{PURIFY}}$ respectively the intervals where the recovery rates of the output banks lie as a function of the recovery rates of the input banks, under any $\epsilon$-\emph{CRRV}. All these intervals either span right from 0 or left from 1. We want to find an assignment \textbf{A} for the value of $\delta$ such that under $m_{\delta}$, the three networks $\mathcal{F}_{\text{NOT}}$, $\mathcal{F}_{\text{OR}}$, and $\mathcal{F}_{\text{PURIFY}}$  \emph{simulate} the gates NOT, OR, and PURIFY respectively, where the notion of simulation is defined in the natural way: for any fixed input bank's recovery rate of the financial network, any $\epsilon$-\emph{CRRV} of the financial network must be mapped back by $m_{\delta}$ to a satisfying assignment on the inputs and outputs of the respective gate.
Consequently, to determine the values of $\delta$ such that the $\mathcal{F}_{\text{NOT}}, \mathcal{F}_{\text{OR}}, \mathcal{F}_{\text{PURIFY}}$ networks correctly simulate the NOT-, OR-, PURIFY-gates respectively, we must set the parameter $\delta$ such that the recovery rates of the output banks of the three networks are contained in the appropriate intervals of  $m_{\delta}$. For example, considering assignments \textbf{A} where the values for $\delta$ satisfy $\Bigl[0,\frac{1}{2}-\delta\Bigr] \subset  \Bigl[0,\frac{1+2\delta}{2\delta}\epsilon\Bigr]$, allows numbers $q\in \Bigl(\frac{1}{2}-\delta,\frac{1+2\delta}{2\delta}\epsilon\Bigr]$ to be mapped according to $m_{\delta}$ to a different value than 0, meaning that $\mathcal{F}_{\text{PURIFY}}$ would not correctly simulate the PURIFY-gate~(Claim \ref{PURIFY-CLAIM}), because for input an encoded value 0, a satisfying assignment for the PURIFY-gate must output encoded values 0 for both output variables. Scanning through all intervals, it turns out the biggest interval is $\Bigl[0,\frac{1+8\delta}{2\delta}\epsilon\Bigr]$ from Statement 2 of Claim \ref{OR-CLAIM}. Eventually to find a  feasible choice of $\delta$, we must consider all assignments \textbf{A} of pairs of values for the parameters $\delta,\epsilon$ where 
\begin{equation}\label{encoding}
\frac{1+8\delta}{2\delta}\epsilon = \frac{1}{2} - \delta.
\end{equation}
As we will show next for all assignments \textbf{A} in which a pair $(\delta,\epsilon)$ satisfies (\ref{encoding}), it holds that all three financial networks correctly simulate their respective gate under $m_{\delta}$. Let $(\delta,\epsilon)$ be any pair that satisfies (\ref{encoding}).

\vspace{3mm}
\begin{itemize}[topsep=0pt, noitemsep]
\item \fbox{\textbf{NOT-gate}} Consider the network $\mathcal{F}_{\text{NOT}}$ of Figure \ref{fig:NOTGATE} and let $r$ be an $\epsilon$-\emph{CRRV}. From Statement 1 of Claim \ref{NOT-CLAIM} whenever $r_{b_u}\in \Bigl[0,\frac{1}{2}-\delta\Bigr]$ it holds that $r_{b_w} \in \Bigl[1-\frac{1+10\delta}{4\delta}\epsilon,1\Bigr]$ and since $\Bigl[1-\frac{1+10\delta}{4\delta}\epsilon,1\Bigr] \subset \Bigl[1-\frac{1+8\delta}{2\delta}\epsilon,1\Bigr]$ trivially $r_{b_w} \in \Bigl[1-\frac{1+8\delta}{2\delta}\epsilon,1\Bigr]$. From Statement 2 of the same Claim, whenever $r_{b_u} \in \Bigl[\frac{1}{2}+\delta,1\Bigr]$, then $r_{b_w} \in \Bigl[0,\frac{1+10\delta}{4\delta}\epsilon\Bigr]$ and since $\Bigl[0,\frac{1+10\delta}{4\delta}\epsilon\Bigr]\subset \Bigl[0,\frac{1+8\delta}{2\delta}\epsilon\Bigr]$ trivially $r_{b_w} \in \Bigl[0,\frac{1+8\delta}{2\delta}\epsilon\Bigr]$. Consequently the mapping $m_\delta$ on any $\epsilon$-\emph{CRRV} of $\mathcal{F}_{\text{NOT}}$ generates an assignment x for which it holds that:
\begin{enumerate}[topsep=0pt, noitemsep]
 \item x$[u]$ = 0 $\rightarrow$ x$[w]$ = 1~(Statement 1)
 \item x$[u]$ = 1 $\rightarrow$ x$[w]$ = 0~(Statement 2)
 \end{enumerate}

 As satisfying assignments are indifferent with respect to $\perp$ values, the above argument suffices to establish that $\mathcal{F}_{\text{NOT}}$ correctly simulates a NOT-gate under $m_\delta$.
 
\vspace{3mm}
\item \fbox{\textbf{OR-gate}} Consider the network $\mathcal{F}_{\text{OR}}$ of Figure \ref{fig:ORGATE} and let $r$ be an $\epsilon$-\emph{CRRV}. Similarly from Statement 1 of Claim \ref{OR-CLAIM}, the recovery rate for the output bank $b_w$ lies in $\Bigl[1-\frac{1+8\delta}{2\delta}\epsilon,1\Bigr]$ if one of $r_{b_u},r_{b_v}$ lies in $\Bigl[\frac{1}{2}+\delta,1\Bigr]$ and from Statement 2, $r_{b_w} \in \Bigl[0,\frac{1+8\delta}{2\delta}\epsilon\Bigr]$ if both $r_{b_u},r_{b_v} \in \Bigl[0,\frac{1}{2}-\delta\Bigr]$. As a result the mapping $m_\delta$ on any $\epsilon$-\emph{CRRV} of $\mathcal{F}_{\text{OR}}$ generates an assignment x for which
\begin{enumerate}[topsep=0pt, noitemsep]
\item if x$[u]$ = x$[v]$ = 0 $\rightarrow$ x$[w]$ = 0~(Statement 2)
 \item if $\text{x}[u] = 1 \text{ or } \text{x}[v] = 1 \rightarrow \text{x}[w] = 1$~(Statement 1)
\end{enumerate}

Checking the satisfying conditions in Section 3.1, we conclude that the above argument suffices to establish that $\mathcal{F}_{\text{OR}}$ correctly simulates the OR-gate under $m_\delta$.

\vspace{3mm}
\item \fbox{\textbf{PURIFY-gate}} Consider the network of Figure \ref{fig:PURIFYGATE}, and let $r$ be an $\epsilon$-\emph{CRRV}.
From Statement 1 of Claim \ref{PURIFY-CLAIM} whenever $r_{b_{u}} \in \Bigl[0,\frac{1}{2}-\delta\Bigr]$ it holds that $r_{b_v} \in \Bigl[0,\frac{1+4\delta}{2\delta}\epsilon\Bigr] \subset \Bigl[0,\frac{1+8\delta}{2\delta}\epsilon\Bigr]$ and $r_{b_w}\in \Bigl[0,\frac{1+2\delta}{2\delta}\epsilon\Bigr] \subset \Bigl[0,\frac{1+8\delta}{2\delta}\epsilon\Bigr]$. That means that the mapping $m_\delta$ generates an assignment x such that if x$[u]$ = 0 then x[$v$] = x[$w$] = 0.  From Statement 2, whenever  $r_{b_u}\in \Bigl[\frac{1}{2}+\delta,1\Bigr]$ then $r_{b_v}\in \Bigl[1-\frac{1+4\delta}{2\delta}\epsilon\Bigr] \subset \Bigl[1-\frac{1+8\delta}{2\delta}\epsilon,1\Bigr]$ and $r_{b_w}\in \Bigl[1-\frac{1+2\delta}{2\delta}\epsilon,1\Bigr] \subset \Bigl[1-\frac{1+8\delta}{2\delta}\epsilon,1\Bigr]$, meaning $m_\delta$ generates an assignment x where if x$[u]$ = 1 then x$[v]$ = x$[w]$ = 1. Finally from Statement 3 and using similar arguments, it is not hard to see that $m_\delta$ generates an assignment x such that whenever x$[u]$ = $\perp$ then x[$v$] = 1 or x$[w]$ = 0. So to conclude, from any $\epsilon$-\emph{CRRV} of $\mathcal{F}_{\text{PURIFY}}$, $m_\delta$ generates an assignment x where:
\begin{enumerate}[topsep=0pt, noitemsep]
\item $\text{x}[u] = 0 \rightarrow \text{x}[v] = \text{x}[w] = 0$~(Statement 1).
 \item  $\text{x}[u] = 1 \rightarrow \text{x}[v] = \text{x}[w] = 1$~(Statement 2).
 \item $\text{x}[u] =\perp \rightarrow \text{x}[v]\in \{0,1\} \text{ or } \text{x}[w] \in \{0,1\}$~(Statement 3).
\end{enumerate}
These are exactly the conditions that hold in a satisfying assignment x of a PURIFY-gate.
\end{itemize}

\vspace{3mm}
Given an instance $I$ of \pure\ that operates on gates NOT, OR and PURIFY, we construct a financial network $\mathcal{F}_I$ consisting of debt and CDS contracts and show that any  $\epsilon$-\emph{CRRV} of $\mathcal{F}_I$ is mapped back to a satisfying assignment x of values to the variables of the original \pure\ instance $I$, by applying $m_\delta$, such that $\delta$ and $\epsilon$ satisfies (\ref{encoding}).

\smallskip \noindent \textbf{The reduction.} Let $I = (V,G)$ be a \pure\ instance.
The construction of the financial network $\mathcal{F}_I$ proceeds as follows: For each gate $g = (\text{NOT},u,w)\in G$, we construct a copy of the financial network $\mathcal{F}_{\text{NOT}}$ of Figure \ref{fig:NOTGATE}. For each gate $g = (\text{OR},u,v,w)\in G$,
 we construct a copy of the financial network $\mathcal{F}_{\text{OR}}$ of Figure \ref{fig:ORGATE}. For each gate $g = (\text{PURIFY},u,v,w)\in G$, we construct a copy of the financial network $\mathcal{F}_{\text{PURIFY}}$ of Figure \ref{fig:PURIFYGATE}. 

The interconnection of the gates in a \pure\ instance rises from the fact that the gates may share variables. Remember though that by the definition of the \pure\ problem~(see Definition \ref{defpure}), no variable can be the output of more than one gate, whereas a variable can be input to many gates. Consequently after the execution of the above steps we might end up with more than one bank to represent the same variable. Next we will show how to deal with these situations.

Without loss of generality, let $\chi$ be the output variable of gate $g$ and the input variable of another gate $g'$. After replacing $g$ and $g'$ with their respective financial networks $\mathcal{F}_g$ and $\mathcal{F}_{g'}$, we are left with two banks representing variable $\chi$, both denoted as $b_{\chi}$. To connect $\mathcal{F}_{g}\text{ and }\mathcal{F}_{g'}$ and represent the interconnection between $g$ and $g'$ due to their shared variable $\chi$, we merge the two $b_{\chi}$ banks into one bank, keeping the same notation $b_{\chi}$. It's worth noting that, as constructed, every input and output bank holds one outgoing liability of notional 1. Therefore, both $b_{\chi}$ banks, which represented the common variable $\chi$ prior to merging, held one outgoing liability of notional 1. In order to maintain the recovery rate analysis presented earlier, it is important for the newly merged $b_{\chi}$ bank to retain only one outgoing liability of notional 1. To achieve this, we eliminate one of the two outgoing liabilities after merging.


We have shown how to generate a financial network $\mathcal{F}_I$ from a \pure\ instance $I$ and presented a mapping $m_\delta$ to map back $\epsilon$-\emph{CRRV}es of $\mathcal{F}_I$ to satisfying assignments for $I$, where $(\delta,\epsilon)$ is a pair that satisfies (\ref{encoding}). The next theorem establishes our main result.

\begin{theorem}\label{main_theorem}
$\epsilon$-\proportional\ is \textsf{PPAD}-hard for $\epsilon \leq \frac{3 -\sqrt 5}{16}$. 
\end{theorem}
\begin{proof}
We reduce from \pure\ to $\epsilon$-\proportional. As established, all values of $\delta,\epsilon$ that satisfy (\ref{encoding}) generate mappings under which each  $\mathcal{F}_{g}$ network correctly simulates gate $g \in \{\text{NOT},\text{OR},\text{PURIFY}\}$. To construct the mapping that maximises the inapproximability parameter $\epsilon$, we rewrite (\ref{encoding}) to $\epsilon = \frac{\delta\cdot(1-2\delta)}{1+8\delta}$ and compute the value $\delta^* = \arg\max \frac{\delta\cdot(1-2\delta)}{1+8\delta}$, where $\delta^*\in (0,1/2)$. Using basic calculus the maximum value for $\epsilon$ is $\frac{3 -\sqrt 5}{16}\approx 0.048$ and is obtained for $\delta^* = \frac{\sqrt 5-1}{8}$. Thus, given an instance $I = (V,G)$ of \pure, we construct a network $\mathcal{F}_I$ of \proportional\ as described above, where we set $\delta = \delta^*$. We can then map, in polynomial time, $\epsilon$-CRRVes of $\mathcal{F}_I$ to a satisfying assignment for $I$ through mapping $m_{\delta^*}$, i.e.:
\begin{itemize}
\item[(i.)] If $r_{b_u} \leq \frac{1}{2}-\delta^*$ then  $\text{x}[u] = 0$;
\item[(ii.)] if $r_{b_u} \geq \frac{1}{2} + \delta^*$ then $\text{x}[u] = 1$
\item[(iii.)] else $\text{x}[u] = \perp$.
\end{itemize}
This establishes the claim.
\end{proof} 

\noindent \textbf{Extension to other payment schemes}.
In the network that we constructed to establish Theorem \ref{main_theorem}, each bank has only one outgoing contract. This allows us to apply the presented analysis to a broader class of payments, including those satisfying \emph{Limited Liability} and \emph{Absolute Priority} conditions.
\begin{corollary}\label{corollary1}
$\epsilon$-\proportional\ is \textsf{PPAD}-hard for $\epsilon \leq \frac{3 -\sqrt 5}{16}$ under any payment scheme that satisfies the Limited Liability and Absolute Priority conditions.
\end{corollary}

\section{The (in)tractability of special financial topologies}\label{central}

This section studies the possibility of efficiently computing exact solutions for certain simpler network structures. At first we point out that clearing vectors of networks that satisfy the \emph{dedicated CDS debtor property} (introduced in \cite{ioannidis2022strong}), are still hard to approximate. Then we introduce the \emph{central CDS debtor} property, where a single bank, assumed to have sufficient assets, is responsible for clearing the network's debt that is generated from the activation of credit default swaps. Despite this being the simplest topology for a financial network with debt and CDS contracts, Theorem \ref{main_theorem} still applies. 

\subsection{Dedicated and central CDS debtors}
In \cite{ioannidis2022strong} the authors defined the \emph{dedicated CDS debtor} property 
and showed that \proportional\ restricted to instances with this property is \textsf{PPAD}-complete (i.e., it is of lower complexity than the general case, which is \textsf{FIXP}-complete).
\begin{definition}[Dedicated CDS debtor~\cite{ioannidis2022strong}]\label{dedicated}
A financial network $\mathcal{F} = (N,e,c)$ satisfies the \emph{dedicated CDS debtor property} iff for every CDS debtor bank $i \in N$, 
\begin{enumerate}[topsep=0pt, noitemsep]
\item  $ \forall j \in N : c_{i,j} = 0$, i.e, there are no debt contracts where $i$ is the debtor.
\item  $\exists R \in N : \forall j,k \in N : c_{i,j}^k \neq 0\rightarrow k = R$, i.e, all CDSes for which $i$ is the debtor share the same reference bank.
\end{enumerate}
\end{definition} 

It is easy 
to verify that the financial networks constructed in the reduction of the previous section satisfy the dedicated CDS debtor property, thus strengthening  the result of \cite{ioannidis2022strong}. 

\begin{corollary}\label{dedicated_theorem}
$\epsilon$-\proportional\ restricted to instances that satisfy the dedicated CDS debtor property is \textsf{PPAD}-hard for $\epsilon \leq \frac{3-\sqrt 5}{16}$.
\end{corollary}

We can show a similar result for another natural restriction on CDS debtors. 
\begin{definition}[Central CDS debtor]\label{central_CDS_debtor}
A financial network $\mathcal{F} = (N,e,c)$ satisfies the central CDS debtor property iff the following conditions hold,
\begin{enumerate}[topsep=0pt, noitemsep]
\item $\exists i' \in N : \forall i,j,R\in N : c_{i,j}^R\neq 0 \rightarrow i = i'$, i.e, all CDS contracts share the same debtor bank. We refer to this debtor as central CDS debtor and denote it by $\mathcal{CCD}$.
\item $\forall i \in N : c_{\mathcal{CCD},i} = 0$, i.e,  there are no debt contracts where $\mathcal{CCD}$ is the debtor.
\item $e_{\mathcal{CCD}} \geq \sum_{j,R \in N} c_{\mathcal{CCD},j}^R$, i.e, $\mathcal{CCD}$ possess enough assets to fully pay off any of its potential liabilities. 
\end{enumerate}
\end{definition}
It turns out that, despite drastic simplification that this property imposes on the way in which CDSes are present in a financial network, it is still not possible to compute a weakly approximate clearing vector efficiently.

\begin{corollary}\label{CCD_theorem}
$\epsilon$-\proportional\ restricted to instances that satisfy the central CDS debtor property is \textsf{PPAD}-hard for $\epsilon \leq \frac{3-\sqrt 5}{16}$.
\end{corollary}

\begin{proof}
It is verified that in the financial networks constructed in the reduction of Theorem \ref{main_theorem}, each CDS debtor satisfies Conditions 2 and 3 of the above definition. We modify the reduction by (as a final step in the construction of the network) merging all CDS debtors into a single debtor node $\mathcal{CCD}$, whose external assets equals the sum of the external assets of merged debtors. The resulting network satisfies \ref{central_CDS_debtor}.
\end{proof}

\subsection{Optimisation-Based Computation of Clearing Vectors}

 On the positive side, we introduce an optimisation-based framework for computing clearing recovery rate vectors, within the context of \emph{central CDS debtors} and the general case. For financial networks that satisfy the \emph{central CDS debtor} property, we present a Mixed-Binary Linear Program (MBLP) which inherently implies an exponential-time algorithm for computing exact clearing recovery rate vectors. Adapting the constraints to general instances, generates a Mixed-Binary Nonlinear Program (MBNLP) for \proportional.
 
\smallskip \noindent \textbf{Mixed-Binary Linear Program for \emph{central CDS debtors}}: Assume a financial network $\mathcal{F} = (N\cup\{\mathcal{CCD}\},e,c)$ that satisfies the \emph{central CDS debtor} property. W.l.o.g assume that each bank in the network apart from $\mathcal{CCD}$ has at least one positive liability\footnote{We make this assumption to avoid some technicalities for the sake of presentation. The program easily adapts to sink nodes.}. We will formulate a Mixed-Binary Linear Program w.r.t $\mathcal{F}$ denoted as $\textbf{MBLP}$($\mathcal{F}$), whose feasible solution set is essentially $\text{Sol}(\mathcal{F})$, i.e, all clearing recovery rate vectors of $\mathcal{F}$. The constraints of the program are formulated w.r.t the variable $z = (r_i,y_i)_{i\in N}$, where $r= (r_i)_{i\in N}$ represents a recovery rate vector with $r_i\in[0,1]$ and $y = (y_i)_{i\in N}$ is a vector of binary decision variables one for each bank $i\in N$ with the following desired indication: For a given recovery rate vector $r$,

\begin{equation}\label{decision variable}
y_i =
\begin{cases}
  0, & \text{if } a_i(r) > l_i \\
  1, & \text{if } a_i(r)< l_i\\
  0 \text{ or } 1 &\text{if } a_i(r) = l_i  
\end{cases}
\end{equation}

Recall that by definition \ref{eq:clearing}, under any clearing vector $r = (r_i)_{i\in N}$, for each bank $i\in N$ it must hold that $r_i = \min(1,a_i(r)/l_i)$, where for the \emph{central CDS debtor} framework,
\begin{equation}\label{assets}
a_i(r) = \Biggl[e_i + \sum_{j\in N}r_j\cdot c_{j,i}+ \sum_{k\in N}(1-r_k)\cdot c_{\mathcal{CCD},i}^k\Biggr]
\end{equation}
\begin{equation}\label{liability}
l_i = \sum_{j\in N}c_{i,j}
\end{equation}
Observe that within networks adhering to the \emph{central CDS debtor} framework, the liabilities of the banks remain fixed, regardless of the recovery rate vector.

The decision vector $y$ serves the purpose of determining for each bank $i\in N$, which argument within the $\min$ operator is the correct choice, whenever the differentiation among the arguments affects the insolvency of bank $i$. Based on such a choice indicated by $y$, for each bank $i\in N$ the program should compute a feasible solution, if a recovery rate vector $r$ that justifies the indicated choice exists. Essentially, $y$ eliminates the need for the $\min$ operator in the fixed-point condition.

Furthermore, we associate each bank $i\in N$, with the following constant 
\begin{equation}\label{beta}
\text{B}_i = \frac{1}{\sum_{j\in N}c_{i,j}}\cdot\Biggl[e_i +\sum_{j\in N}c_{j,i} + \sum_{k\in N}c^k_{\mathcal{CCD},i}\Biggr] + 1.
\end{equation}

Clearly for each bank $i\in N$ and any recovery rate vector $r=(r_i)_{i\in N}$,  $\min(1,a_i(r)/l_i)\leq \text{B}_i$.

\vspace{2mm}
We proceed by presenting a Mixed-Binary Linear Program for computing exact clearing recovery rate vectors for a given financial network $\mathcal{F} = (N\cup\{\mathcal{CCD}\},e,c)$ that satisfies the \emph{central CDS debtor} property. At this stage we exclusively address the feasibility problem of just computing any clearing vector.  Therefore, the program is presented without specifying an objective function, rather we assume any linear function $f(r)$, w.r.t a recovery rate vector $r = (r_i)_{i\in N}$.

\begin{figure}[htb]
  \centering
\begin{align}
\setcounter{equation}{0}
\text{Maximise / Minimise:} \quad & \text{Linear function } f(r) \nonumber\\
\text{Subject to:} \quad & \text{For each } i\in N \nonumber\\
& r_i \geq \frac{a_i(r)}{l_i} - \text{B}_i\cdot (1 - y_i) \tag{constraint 1}\label{con1}\\
& r_i \geq 1 - \text{B}_i\cdot y_i  \tag{constraint 2}\label{con2}\\
& r_i \leq \frac{a_i(r)}{l_i} \tag{constraint 3}\label{con3}\\
& r_i \in [0,1] \tag{constraint 4}\label{con4}\\
& y_i\in\{0,1\} \tag{constraint 5}\label{con5} 
\end{align}
\caption{The Mixed-Binary Linear Program for \emph{central CDS debtors}.}
  \label{MBLP_1}
\end{figure}
  
\begin{theorem}\label{MBLP}
 Assume a financial network $\mathcal{F} = (N\cup\{\mathcal{CCD}\},e,c)$ that satisfies the central CDS debtor property and construct the Mixed-Binary Linear Program in Figure \ref{MBLP_1} w.r.t. $\mathcal{F}$ denoted as $\textbf{MBLP}(\mathcal{F})$.
 \begin{enumerate}
 \item[i)] If $z=(r_i,y_i)_{i\in N}$ is a feasible solution of $\textbf{MBLP}(\mathcal{F})$ then $(r_i)_{i\in N}$ is a clearing vector of $\mathcal{F}$. 
 \item[ii)] If $(r_i)_{i\in N}$ is a clearing vector of $\mathcal{F}$ then there exist a binary vector  $(y_i)_{i\in N}$ s.t $z=(r_i,y_i)_{i\in N}$ is a feasible solution of $\textbf{MBLP}(\mathcal{F})$. 
\end{enumerate}
\end{theorem}

 \begin{proof}
Assume the $\textbf{MBLP}(\mathcal{F}$) of Figure \ref{MBLP_1}. The program's constraints contain a mixture of real valued variables $r_i\in[0,1]$ and binary valued decision variables $y_i\in \{0,1\}$ for each bank $i\in N$. The linearity of all constraints w.r.t $(r_i)_{i\in N}$ and $(y_i)_{i\in N}$ is justified from equations \ref{assets}, \ref{liability} and \ref{beta}.

\emph{i)} Let $z = (r_i,y_i)_{i\in N}$ be a feasible solution to \ref{MBLP_1} and w.l.o.g fix an index $\kappa$. If  $y_{\kappa} = 0$ then from \ref{con2} and \ref{con4} we get that $r_{\kappa} = 1 $, while subsequently from \ref{con3} it holds that $1\leq a_{\kappa}(r)/l_{\kappa}$. Note that \ref{con1} is trivially satisfied due to the choice of B$_{\kappa}$. If $y_{\kappa} = 1$, then from \ref{con1} and \ref{con3} it holds that $r_{\kappa} = a_{\kappa}(r)/l_{\kappa}$, while combining the latter with \ref{con4} implies that $a_{\kappa(r)}/l_{\kappa}\leq 1$. Note that \ref{con2} is trivially satisfied due to the choice of B$_{\kappa}$. Consequently for both possible values of the binary variable $y_{\kappa}$, we showed that the satisfied constraints imply that $r_{\kappa} = \min(1,a_{\kappa}(r)/l_{\kappa})$, thus establishing that the fixed point condition of definition \ref{eq:clearing} is satisfied by $(r_i)_{i\in N}$. Therefore if $z = (r_i, y_i)_{i \in N}$ is a feasible solution of $\textbf{MBLP}({\mathcal{F})}$ then $(r_i)_{i \in N}$ is a clearing vector of $\mathcal{F}$.

\emph{ii)} Let $(r_i)_{i\in N}$ be a clearing vector of $\mathcal{F}$ and w.l.o.g fix an index $\kappa$. We will prove that the sole arrangement for $y_{\kappa}$ that satisfies the constraints for $\kappa$, corresponds to the configuration implied by \ref{decision variable} w.r.t $(r_i)_{i\in N}$. Given that $(r_i)_{i\in N}$ is assumed to be a clearing vector, from definition \ref{eq:clearing}, it must hold that $r_{\kappa} = \min(1,a_{\kappa}(r)/l_{\kappa})$. If $\min(1,a_{\kappa}(r)/l_{\kappa}) = 1$, then $r_{\kappa} = 1$. By setting $y_{\kappa} = 0$, both \ref{con2} and \ref{con4} are satisfied, while by assumption $1\leq a_{\kappa}(r)/l_{\kappa}$ thus \ref{con3} is also satisfied. Trivially \ref{con1} is satisfied under this choice of $y_{\kappa}$ due to B$_{\kappa}$. If $\min(1,a_{\kappa}(r)/l_{\kappa}) = a_{\kappa}(r)/l_{\kappa}$, then $r_{\kappa} = a_{\kappa}(r)/l_{\kappa}$. Consequently by setting $y_{\kappa} = 1$, both \ref{con1} and \ref{con3} are satisfied, while \ref{con4} is also satisfied since by assumption $a_{\kappa}(r)/l_{\kappa}\leq 1$. Trivially \ref{con2} is satisfied under this choice for $y_{\kappa}$. In summary we showed that for a given clearing vector $(r_i)_{i\in N}$ of $\mathcal{F}$, setting $(y_i)_{i\in N}$ according to configuration \ref{decision variable} constitutes $z = (r_i,y_i)_{i\in N}$ a feasible solution for $\textbf{MBLP}(\mathcal{F})$.
\end{proof}

If we could determine which configurations of the binary decision variable vector $y =  (y_i)_{i\in N}$ can generate a clearing vector $r = (r_i)_{i\in N}$, then computing $r$ simply comes down to solving a Linear Program. As a result, a direct consequence of the proposed Mixed-Binary Linear Program is an exponential-time algorithm for computing clearing recovery rate vectors for networks meeting the \emph{central CDS debtor} property.

\begin{theorem}\label{exp_alg}
\proportional\ restricted to instances that satisfy the central CDS debtor property  admits an exponential time algorithm.
\end{theorem}
\begin{proof}
Consider a financial network $\mathcal{F} = (N\cup\{\mathcal{CCD}\},e,c)$, that satisfies the \emph{central CDS debtor} property and construct  the \textbf{MBLP}($\mathcal{F}$) as indicated in Figure \ref{MBLP_1}. Based on the structure of \textbf{MBLP}($\mathcal{F}$) and the preceding  discussion that states the linearity of the constraints, it is evident that when fixed values for the binary decision variable vector $y$ are introduced, the formulation of Figure \ref{MBLP_1} transforms into a Linear Program w.r.t to the real variable vector $r = (r_i)_{i\in N}$. We denote the generated LP for a fixed vector $y$ as \textbf{LP}$(\mathcal{F}\lvert y$). The algorithm iterates over all configurations of $y \in \{0,1\}^{\lvert N \vert}$ one at a time and invokes any polynomial time algorithm designed for solving linear programs as a subroutine for solving \textbf{LP}($\mathcal{F}\lvert y$). In case \textbf{LP}($\mathcal{F}\lvert y$) is feasible the subroutine algorithm will compute and return a clearing recovery rate vector $r = (r_i)_{i\in N}$ and the algorithm will terminate. In case \textbf{LP}($\mathcal{F}\lvert y$) is infeasible the subroutine algorithm returns nothing, the algorithm considers a new unprocessed configuration for $y$ and repeats the execution of the steps described so far. The algorithm is illustrated below. 

\begin{algorithm}[h]
   \caption{An exponential time algorithm for computing a clearing vector for \emph{central CDS debtors}.\label{exptimealgo}}
    \begin{algorithmic}[1]
        \STATE  Let $\mathcal{F} = (N\cup \{\mathcal{CCD}\},e,c)$ be the input network that satisfies the \emph{central CDS debtor} property.
        \STATE Construct the Mixed-Binary Linear Program for $\mathcal{F}$ as described in Figure \ref{MBLP_1}.
        \FOR{$y \in \{0,1\}^{\mid N\mid}$}
        \IF{ LP($\mathcal{F}\lvert y$) is feasible }
        \STATE Return the feasible point.
        \STATE Break.
        \ENDIF
        \ENDFOR
    \end{algorithmic}
\end{algorithm}

As \proportional\ is a total search problem the existence of a clearing recovery rate vector is guaranteed. Therefore Algorithm \ref{exptimealgo} is guaranteed to output a clearing recovery rate vector. In the worst case the algorithm would have to execute a polynomial time subroutine on all $2^{\lvert N \lvert}$ possible configurations for vector $y$.
This implies a running time of $\mathcal{O}(poly\lvert\mathcal{F}\lvert)\cdot2^{\lvert N\lvert}$, where $\vert \mathcal{F}\lvert$ is the bit length of the input.
\end{proof}
\begin{remark}\label{expremark}
We intentionally avoided specifying a particular objective function. This deliberate choice was made to show the framework's versatility in optimising over a wide range of linear objectives tied to the clearing vector, which results, in an algorithmic scheme for computing the clearing vector, that effectively optimises any given linear objective function. Many concepts on proposed objective functions of interest have been addressed in the literature, but primarily from a computational hardness standpoint \cite{papp2022default,DBLP:journals/tcs/IoannidisKV23}.
In \cite{papp2022default}, the authors address ambiguity of CRRVs by highlighting how the multiplicity of CRRVs gives rise to optimisation problems whereby one attempts to select the appropriate clearing vector that optimises a desired linear objective. They prove a set of \textsf{NP}-hardness results regarding the choice of the clearing vector $r$ that can satisfy certain objectives expressed as linear functions in $r$. The reductions therein are based on high capitalised CDS debtors\footnote{CDS debtors with no debt contracts and $\infty$ external assets.} that can be merged into a single central CDS debtor. As an application, our algorithm could be used to resolve this ambiguity for linear objective functions when central CDS debtors are present.
 
Moreover, despite the exponential time complexity of the algorithm, the procedure is exponential solely in the number of banks in the network: The size of the coefficients in the input (i.e., contract notionals and external assets) have no impact on the exponent.

\end{remark}
\smallskip \noindent \textbf{Mixed-Binary Nonlinear Program for \proportional\ }:  At this point we mention that adapting  the assets, liabilities, the constants $(\text{B}_i)_{i\in N}$ and the binary decision variable $y$ to the expressions
\begin{equation*}\label{assets_2}
a_i(r) = \Biggl[e_i + \sum_{j\in N} r_j\cdot c_{j,i}+ \sum_{j,k\in N}r_j\cdot(1-r_k)\cdot c_{j,i}^k\Biggr]
\end{equation*}
\begin{equation*}\label{liability_2}
l_i(r) = \sum_{j,k\in N}(1-r_k)\cdot c_{i,j} +\sum_{j\in N}c_{i,j}
\end{equation*}
\begin{equation*}\label{newB}
\text{B}_i =\frac{1}{\sum_{j\in N}c_{i,j}}\cdot\Biggl[ e_i +\sum_{j\in N}c_{j,i} + \sum_{j,k\in N}c^k_{j,i}\Biggr] + 1
\end{equation*}
we form a Mixed-Binary Nonlinear Program denoted as \textbf{MBNLP}($\mathcal{F}$) tailored for general instances $\mathcal{F}$ of \proportional\ .
 
\begin{remark}
Constructing heuristic methods for \proportional\ has not been addressed in the literature. The construction of the Mixed-Binary Nonlinear Program for general instances of \proportional\ suggests heuristic-based approaches and originates the empirical study of the problem. 
\end{remark}

\subsection{Positive results for central CDS debtors with covered CDSes}
Section 4.1 identifies two severe restrictions of $\epsilon$-\proportional\ under which the inapproximability result of Theorem \ref{main_theorem} remains to hold. Nonetheless, we would like to identify non-trivial and important special cases of \proportional\ whose solutions admit efficient algorithms. It turns out that this is achieved when we insist on the use of \emph{covered} CDSes  (a notion defined in \cite{schuldenzucker2016clearing}) on top of central CDS debtors~\ref{central_CDS_debtor}. A CDS is said to be covered if there exists a debt contract from the reference bank to the creditor of the CDS, with a notional that exceeds the CDS notional. 
\begin{definition}[Covered CDS]\label{covered_cds} A credit default swap $(i,j,R)$ is \emph{covered} iff $c_{i,j}^R \leq c_{R,j}$.
\end{definition}

Covered CDSes form a common occurence in practice, as using a CDS in this way provides a form of insurance against a debt contract with an insolvent debtor. We provide a polynomial-time algorithm for computing clearing vectors for networks with the \emph{central CDS debtor} property and \emph{covered CDS contracts}. 

\begin{theorem}\label{positive}
\proportional\ restricted to instances that only contain covered CDSes and satisfy the central CDS debtor property admits a polynomial time algorithm.
\end{theorem}
To prove this, first we perform a transformation step on each credit default swap in the given network, which locally modify the network. Repeating this process on every CDS contract results in a financial network consisting only of debt contracts, which we can then solve through the polynomial-time algorithm presented in \cite{eisenberg2001systemic} that computes an exact \emph{CRRV}. 

\smallskip \noindent \textbf{Network transformation.} 
Let $\mathcal{F} = (N\cup \mathcal{CCD},e,c)$ be a financial network that satisfies the \emph{central CDS debtor} property where all CDS contracts are \emph{covered}. Consider a CDS contract $(\mathcal{CCD},j,R)$ in the network, let $x = c_{\mathcal{CCD},j}^R$, let $y = c_{R,j}$ and let $k$ be such that $y = x + k$. A \emph{network transformation step} on $(\mathcal{CCD},j,R)$ consists of the following consecutive operations.
\begin{enumerate}[topsep=0pt, noitemsep]
    \item Update the external assets of $j$ to  $e^*_j = e_j + x$;
    \item Decrease the contract notional of the debt contract $(R,j)$ to $c^*_{R,j} = k$;
    \item Add a dummy node, which we call $\text{dummy}_{(R,j)}$;
    \item Add a debt contract $(R,\text{dummy}_{(R,j)})$ with contract notional $c_{R,\text{dummy}_{(R,j)}} = x$;
    \item Erase $(\mathcal{CCD},j,R)$. 
\end{enumerate}
Figure \ref{COVERED_CCP} illustrates this transformation step.

\begin{figure}[htbp!]
\centering
\scalebox{0.8}   
    {\begin{tikzpicture}
[shorten >=1pt,node distance=2cm,initial text=]
\tikzstyle{every state}=[draw=black!50,very thick]
\tikzset{every state/.style={minimum size=0pt}}
\tikzstyle{accepting}=[accepting by arrow]
\node(1){$R$};  
\node(2)[below left of =1,yshift=-5mm]{$\mathcal{CCD}$}; 
\node(3)[below right of =1,yshift=-5mm]{$j$}; 
\draw[orange,->,very thick,snake=snake] (2)--node[midway,black,yshift=-2mm]{$x$}(3);
\draw[blue,->,very thick] (1)--node[midway,black,xshift=1cm,yshift=2mm]{$y = x + k$}(3);
\path [orange,->,draw,dashed,thick] (1) -- ($ (2) !.5! (3) $);
\node[teal,below of=3,yshift=1.5cm]{$e_j$};

\end{tikzpicture}
\qquad\tikz[baseline=-\baselineskip]\draw[ultra thick,->,yshift=0.8cm] (1,0) -- ++ (1,0);
\qquad
{
\begin{tikzpicture}
[shorten >=1pt,node distance=2cm,initial text=]
\tikzstyle{every state}=[draw=black!50,very thick]
\tikzset{every state/.style={minimum size=0pt}}
\tikzstyle{accepting}=[accepting by arrow]
\node(1){$R$};  
\node(4)[right of=1]{$\text{dummy}_{(R,j)}$};
\node(3)[below right of =1,yshift=-5mm]{$j$}; 
\draw[blue,->,very thick] (1)--node[midway,black,yshift=2mm]{$x$}(4);
\node[teal,below of=3,yshift=1.5cm]{$e^*_{j} = e_i + x$};
\node(5)[left of =3]{$\mathcal{CCD}$};
\draw[blue,->,very thick](1)--node[midway,black,xshift=4mm]{$k$}(3);
\end{tikzpicture}}}
\caption{The reconfiguration of the dynamics after the removal of the covered CDS.}
    \label{COVERED_CCP}
\end{figure}
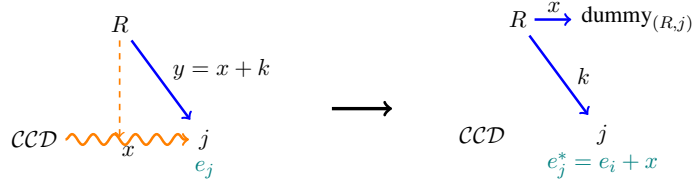
Let $\mathcal{F}$ be a financial network with the central CDS debtor property and consider the execution of a single transformation step on a covered CDS  $(\mathcal{CCD},j,R)$ of $\mathcal{F}$. Let $r$ be a CRRV of $\mathcal{F}$, and let $\mathcal{F'}$ be the financial network after the transformation step. We define $r'$ as the recovery rate vector of $\mathcal{F'}$ that is obtained from $r$ by letting $r'_{\text{dummy}(R,j)} = 1$ and letting $r'$ coincide with $r$ for all other banks. Note that, by construction, the total liabilities of $R$ are equal in $\mathcal{F}$ and $\mathcal{F}'$, and $j$'s assets in $\mathcal{F}$ under $r$ are equal to $j'$'s assets in $\mathcal{F}'$ under $r'$. Hence, $r'$ is a CRRV for $\mathcal{F}'$ under which every node's assets and liabilities remain unaffected with respect to $\mathcal{F}$ under $r$. 

Thus, for a financial network that has only covered CDSes and satisfies the central CDS debtor property, it is possible to obtain an equivalent financial network without CDSes by repeatedly executing the above transformation step on each of the network's CDSes. The resulting CDS-free network then has essentially the same clearing recovery rate vectors as the original network, where the only difference is that the newly introduced dummy nodes always get assigned a recovery rate of $1$.

A CRRV for the resulting CDS-free network can be found using e.g. the polynomial time algorithm of \cite{eisenberg2001systemic}, hence by throwing away from that CRRV the coordinates of the introduced dummy nodes, we obtain a CRRV for the original network in polynomial time. 
This procedure is summarised as Algorithm \ref{polytimealgo} below.

\begin{algorithm}[h]
   \caption{A polynomial time algorithm for computing a CRRVs with central CDS debtors  and covered CDSes.\label{polytimealgo}}
    \begin{algorithmic}[1]
        \STATE  Let $\mathcal{F} = (N\cup \{\mathcal{CCD}\},e,c)$ with $N = [n]$ be the input network with only covered CDSes, and the central CDS debtor property.
        \FOR{$j$ = $1$ to $[n]$}
        \FOR{$R$ = $1$ to $[n]$}
        \IF{ $c_{\mathcal{CCD},j}^R\neq 0$}
        \STATE $e_j = e_j + c_{\mathcal{CCD},j}^R$
        \STATE $c_{R,j} = c_{R,j} - c_{\mathcal{CCD},j}^R $
        \STATE $N = N \cup \text{dummy}_{R,j}$
        \STATE $c_{R,\text{dummy}_{R,j}} = c_{\mathcal{CCD},j}^R $
        \STATE $c_{\mathcal{CCD},j}^R = 0$
        \ENDIF
        \ENDFOR
        \ENDFOR
         \STATE Run the polynomial time algorithm of \cite{eisenberg2001systemic} to obtain a CRRV $r$, and return $r$ restricted to $N \cup \{\mathcal{CCD}\}$.
    \end{algorithmic}
\end{algorithm}

 From Lines 2 through 12 the algorithm repeatedly executes the network transformation step, consisting of Operations 1 through 5, on all CDS contracts of the form $(\mathcal{CCD},j,R)$ of the input financial network. This results in a network $\mathcal{F}'$ that is composed entirely of debt contracts. The final line of the algorithm from Line 13 runs the polynomial time algorithm of \cite{eisenberg2001systemic} on $\mathcal{F}'$ to compute the clearing recovery rate vector, and throws away the coordinates corresponding to the introduced dummy nodes.

Lastly, we note that this polynomial time procedure also works correctly under a generalised version of the \emph{central CDS debtor} property, where we allow multiple banks to be debtors of CDSes, but still require that these CDSes have an amount of external assets exceeding the sum of the notionals, akin to Conditions 2 and 3 of Definition \ref{central_CDS_debtor}. This reflects a setting where banks are risk averse and refuse (or are only authorised) to act as debtor in a CDS when it is certain that all liabilities resulting from these CDSes can be paid off through their assets a-priori.

\section{Discussion and Future work}
We consider this work fundamental for future theoretical study on connections among financial networks and computation, as well as for empirical and heuristic-based methods oriented towards financial applications.
Numerous open questions and potential avenues for future research pop up. Whether the bound of Theorem \ref{main_theorem} is tight is the main open question. Our intuition is that for higher bounds the financial gate analysis should induce less errors. This is achieved by reducing the number of CDS contracts in the gates, which in our opinion is challenging if not impossible.
\begin{conjecture}
Theorem \ref{main_theorem} is tight for central CDS debtor instances.
\end{conjecture}The programs presented in section 4.2 are a solid base for stronger positive results. A compelling challenge is to devise a polynomial-time approximation algorithm for $\epsilon$-\proportional\ and gain intuition on how close to the inapproximability parameter we can get. In light of the community's growing interest in exploring the smoothed complexity of problems associated with the classes \textsf{PLS} but also \textsf{PPAD}, another research direction involves studying the smooth complexity of \proportional. An underline connection could possibly be drawn between the program presented in 4.2 and the smoothed analysis of integer programming \cite{DBLP:journals/mp/RoglinV07}.

In absence of pure algorithmic methods, the empirical investigation of the problem can be initiated by referring to Section 4.2. This could contribute in clearing mechanisms and heuristic methods for monitoring systemic risk in complex networks as well as providing a better understanding of the problem itself. An interesting challenge would be devising machine learning based algorithms or methods to tackle the problem. Clearing mechanisms that utilise reinforcement learning algorithms trained for this purpose would not only pique the interest of the financial community but also garner attention from the theoretical machine learning community.

\bibliographystyle{alpha}
\bibliography{bibl.bib}

\appendix
\section{Example of a financial network}\label{apx:A}

The financial system of Figure \ref{fig:1} consists of six banks, $N = \{1,2,3,4,5,6\}$. Banks 2 and 5 have external assets $e_2 = e_5 = 1-c$, for some constant $c \in (0,1)$,  
while all other banks have zero external assets. The set of debt contracts is $\mathcal{DC} = \{(2,3),(5,4)\}$ and the set of CDS contracts is $\mathcal{CDS} = \{(2,1,5),(5,6,2)\}$. All contract notionals are set to 1. For example, $c_{2,3} = c_{2,1}^5 =1$.

\noindent{\bf{Clearing recovery rate computation.}} By \eqref{eq:clearing} it must be $r_2 =  \min \left\{1,a_2(r)/l_2(r) \right\}$ and $r_5 =  \min \left\{1,a_5(r)/l_5(r) \right\}$. For node 2 its liability is $l_2 = l_{2,3} + l_{2,1} = c_{2,3} + c_{2,3}^5 = 1 + ( 1 - r_5) = 2-r_5$ and symmetrically for node 5 it holds $l_5 = l_{5,6} + l_{5,4} = 1-r_2$. For the assets, since no payment incomes neither nodes 2 and 5 it holds that $a_2 = e_2 = 1-c$ and also $a_5 = e_5 = 1-c$. So $r_2 = \min\{1,a_2(r)/l_2(r)\} = \min\{1,\frac{1-c}{2-r_5}\}$ and symmetrically $r_5 = \min\{1,\frac{1-c}{2-r_2}\}$. To compute the clearing recovery rate we have to solve the system of the two equations, which evetually comes down to solving the system $r_2 = \frac{1-c}{2-\frac{1-c}{2-r_2}} \rightarrow r_2^2-2r_2 + 1-c = 0 \rightarrow r_2 = 1-\sqrt c$. So $r_2 = \min\{1,1-\sqrt c\}$ and symmetrically $r_5 = \min\{1,1-\sqrt c\}$. Depending on the value of $c$ we get different values for the clearing recovery rates of 2 and 5. If $c = 1/4$ then $r_2 = r_5 = 1/2$.
In the proportional payment scheme, each bank pays proportionally to its recovery rate. The payment of node 2 towards node 3 is $p_{2,3} = r_2l_{2,3} = 1-\sqrt c$ and $p_{2,1} = r_2l_{2,1}^5  = (1-\sqrt c)(1-r_5) = (1-\sqrt c)\sqrt c$. So finally $p_{2}(r) = p_{2,1}(r)+p_{2,3}= 1-\sqrt c + (1-\sqrt c)\sqrt c$. The payments for node 5 are symmetrical.

\begin{figure}[htbp!]
   \centering
\scalebox{1}    
{\begin{tikzpicture}
[shorten >=1pt,node distance=2cm,initial text=]
\tikzstyle{every state}=[draw=black!50,very thick]
\tikzset{every state/.style={minimum size=0pt}}
\node (1) {$1$}; 
\node (2) [right of=1] {$2$};
\node (3) [right of=2] {$3$};

\draw[orange,very thick,->,snake=snake] (2)--node[midway,black,yshift=2.5mm]{1}(1);
\draw[blue,very thick,->] (2)--node[midway,black,yshift=2mm]{1}(3);

\node(7) [below of=2] {$5$};
\node(8) [right of=7] {$6$};
\node(6) [left of=7]  {$4$};
\draw[blue,very thick,->] (7)--node[midway,black,yshift=-3mm]{1}(6);
\draw[orange,very thick,->,snake=snake] (7)--node[midway,black,yshift=-3mm]{1}(8);
\node[teal,right of=2,xshift=-2cm,yshift=5mm]{$1-c$};
\node[teal,right of=7,xshift=-2cm,yshift=-6mm]{$1-c$};
\path [orange,-,draw,dashed,thick] (7) -- ($ (2) !.5! (1) $);
\path [orange,-,draw,dashed,thick] (2) -- ($ (7) !.5! (8) $);
\end{tikzpicture}}
 \caption{Example of a financial network.}
    \label{fig:1}
\end{figure}
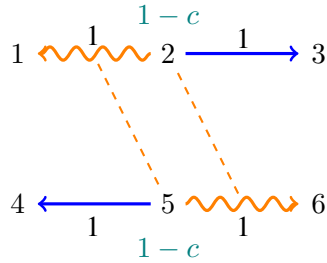

\section{Illustration, example and simulation of PURE-CIRCUIT}\label{purecircuitinstance}

We illustrate each gate $g$ by a directed graph that contains one corresponding node for each variable in $g$ and a special node that indicates the type of the gate. Nodes standing for input variables have an directed edge towards the special node who itself has a directed edge towards every node that represents an output variable. The graphs are presented below.

\begin{figure}[htbp!]
    \centering
\begin{tikzpicture}[shorten >=1pt,node distance=1.5cm,initial text=]
\tikzstyle{every state}=[draw=black!50,very thick]
\tikzset{every state/.style={minimum size=0pt}}
\tikzstyle{accepting}=[accepting by arrow]
\node(1){$u$};  
\node(2)[below  of =1]{$\text{NOT}$}; 
\node(3)[below of =2]{$w$}; 
\draw[black,->] (1)--(2);
\draw[black,->] (2)--(3);
\end{tikzpicture}
\qquad
\begin{tikzpicture}[shorten >=1pt,node distance=1.5cm,initial text=]
\tikzstyle{every state}=[draw=black!50,very thick]
\tikzset{every state/.style={minimum size=0pt}}
\tikzstyle{accepting}=[accepting by arrow]
\node(1){$u$};  
\node(4)[right of=1,xshift=5mm]{$v$};
\node(2)[below  of =1,xshift=1cm]{$\text{OR}$}; 
\node(3)[below of =2]{$w$}; 
\draw[black,->] (1)--(2);
\draw[black,->] (2)--(3);
\draw[black,->] (4)--(2);
\end{tikzpicture}
\qquad
\begin{tikzpicture}[shorten >=1pt,node distance=1.8cm,initial text=]
\tikzstyle{every state}=[draw=black!50,very thick]
\tikzset{every state/.style={minimum size=0pt}}
\tikzstyle{accepting}=[accepting by arrow]
\node(1){$u$};  
\node(2)[below  of =1]{$\text{PURIFY}$}; 
\node(3)[below right of =2]{$w$}; 
\node(4)[below left of=2]{$v$};
\draw[black,->] (1)--(2);
\draw[black,->] (2)--(3);
\draw[black,->] (2)--(4);
\end{tikzpicture}
\caption{Graphical representation for the NOT, OR, PURIFY gates}
    \label{Gates}
\end{figure}
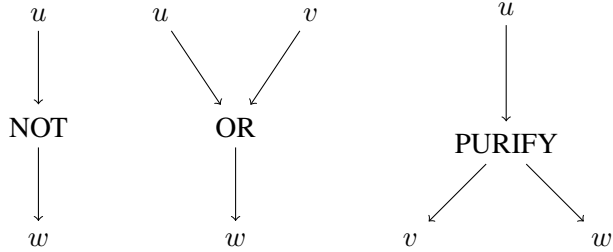

As an example, consider the instance $I = (V,G)$ with $V = \{u,v,w,y\}$ and  $G = \{(\text{NOT},u,v),(\text{OR},v$ $,w,y), (\text{PURIFY},v,u,w)\}$. No variable is output to more than one gate, which constitutes $I$ a valid \pure\ instance. 
In any satisfying assignment x, neither $\text{x}[u] \text{ nor }\text{x}[v]$ can lie in $\{0,1\}$, thus it must hold that x$[v] = \text{x}[u]  = \perp$ and $\text{x}[w] \neq \perp$, otherwise the PURIFY gate is not satisfied. If we set $\text{x}[w] = 1$ then also $\text{x}[y] = 1$ due to the OR-gate. Therefore, a satisfying assignment is $\text{x}[u] = \text{x}[v] = \perp$ and $\text{x}[w] = \text{x}[y] = 1$
The graphical illustration of $I$ is presented below. 

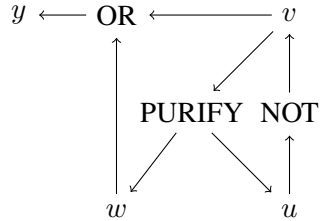
\begin{figure}[htbp!]   
    \centering
    \begin{tikzpicture}
[shorten >=1pt,node distance=1.3cm,initial text=]
\tikzstyle{every state}=[draw=black!50,very thick]
\tikzset{every state/.style={minimum size=0pt}}
\tikzstyle{accepting}=[accepting by arrow]
\node(1){$v$};  
\node(2)[below  of =1]{$\text{NOT}$}; 
\node(3)[below of =2]{$u$}; 
\draw[black,->] (2)--(1);
\draw[black,->] (3)--(2);
\node(4)[left of =2]{$\text{PURIFY}$};
\draw[black,->] (1)--(4);
\draw[black,->] (4)--(3);
\node(5)[left of=3,xshift=-1cm]{$w$};
\draw[black,->] (4)--(5);
\node(6)[above of=4,xshift=-1cm]{$\text{OR}$};
\draw[black,->] (1)--(6);
\draw[black,->] (5)--(6);
\node(7)[left of=6]{$y$};
\draw[black,->] (6)--(7);
\end{tikzpicture}
\caption{The graphical illustration of $I$.}
    \label{purecircuitinstance_I}
\end{figure}
Executing the steps presented in the reduction we can construct the financial network $\mathcal{F}_I$ that simulates $I$. Observe that  each variable $x$ is represented by a bank $b_x$. Specifically since variable $v$ is input to two gates, bank $b_v$ is reference to three CDS contracts two belonging in $\mathcal{F}_{\text{PURIFY}}$ and one in $\mathcal{F}_{\text{OR}}$. Since no variable $x$ is output to more than one gate, there exists no bank $b_x$ with more than one outgoing blue arc. Finally as pointed above, all CDS debtors posses enough assets to pay all generated liabilities and there exist no bank with more than one outgoing arc. The network is presented below where $\delta$ is defined in section 3.2.

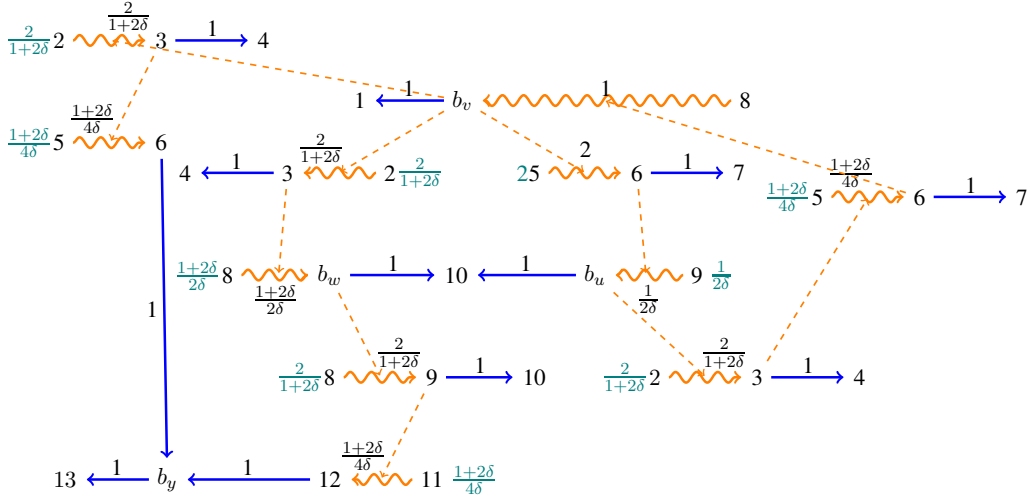
\begin{figure}[htbp!]
  
    \centering
   
\scalebox{0.8}
{
\begin{tikzpicture}
[shorten >=1pt,node distance=1.7cm,initial text=]
\tikzstyle{every state}=[draw=black!50,very thick]
\tikzset{every state/.style={minimum size=0pt}}
\tikzstyle{accepting}=[accepting by arrow]

\node(1) {$b_v$};
\node(2)[left of=1]{1};
\draw[blue,->,very thick] (1)--node[midway,black,yshift=2mm]{1}(2);

\node(3)[below right of=1]{5};
\node[teal,left of=3,xshift=1.5cm]{2};
\node(4)[right of=3]{6};
\node(5)[right of=4]{7};
\draw[blue,->,very thick] (4)--node[midway,black,yshift=2mm]{1}(5);
\draw[orange,->,very thick,snake=snake] (3)--node[midway,black,yshift=4mm]{2}(4);

\node(6)[below left of=1]{2};
\node(7)[left of=6]{3};
\draw[orange,->,very thick,snake=snake] (6)--node[midway,black,yshift=4mm,xshift=-3mm]{$\frac{2}{1+2\delta}$}(7);
\node(8)[left of=7]{4};
\draw[blue,->,very thick] (7)--node[midway,black,yshift=2mm]{1}(8);
\node[teal,right of=6,xshift=-1.2cm]{$\frac{2}{1+2\delta}$};
\path [orange,->,draw,dashed,thick] (1) -- ($ (3) !.5! (4) $);
\path [orange,->,draw,dashed,thick] (1) -- ($ (6) !.5! (7) $);

\node(9)[below of=4,xshift=1cm]{9};
\node[teal,right of=9,xshift=-1.3cm]{$\frac{1}{2\delta}$};
\node(10)[left of=9]{$b_u$};
\draw[orange,->,very thick,snake=snake] (9)--node[midway,black,yshift=-4mm]{$\frac{1}{2\delta}$}(10);
\node(11)[below of=7,xshift=-1cm]{8};
\node[teal,left of=11,xshift=1.2cm]{$\frac{1+2\delta}{2\delta}$};
\node(12)[right of=11]{$b_w$};
\draw[orange,->,very thick,snake=snake] (11)--node[midway,black,yshift=-4mm]{$\frac{1+2\delta}{2\delta}$}(12);
\node(13)[right of=12,xshift=4mm]{10};

\draw[blue,->,very thick] (12)--node[midway,black,yshift=2mm]{1}(13);
\draw[blue,->,very thick] (10)--node[midway,black,yshift=2mm]{1}(13);
\path [orange,->,draw,dashed,thick] (7) -- ($ (11) !.5! (12) $);
\path [orange,->,draw,dashed,thick] (4) -- ($ (9) !.5! (10) $);


\node(3)[below of=10,xshift=1cm]{$2$};
\node[teal,left of=3,xshift=1.2cm]{$\frac{2}{1+2\delta}$};
\node(4)[right of=3]{3};
\node(5)[right of=4]{4};
\draw[blue,->,very thick] (4)--node[midway,black,yshift=2mm]{1}(5);
\draw[orange,->,very thick,snake=snake] (3)--node[midway,black,yshift=4mm,xshift=3mm]{$\frac{2}{1+2\delta}$}(4);
\node(6)[right of =3,,xshift=1cm,yshift=3cm]{5};
\node(7)[right of=6]{6};
\draw[orange,->,very thick,snake=snake](6)--node[midway,black,yshift=4mm,xshift=-3mm]{$\frac{1+2\delta}{4\delta}$}(7);
\node[teal,left of=6,xshift=1.2cm]{$\frac{1+2\delta}{4\delta}$};
\node(8)[right of=7]{7};
\draw[blue,->,very thick] (7)--node[midway,black,yshift=2mm]{1}(8);
\path [orange,->,draw,dashed,thick] (4) -- ($ (6) !.5! (7) $);

\path [orange,->,draw,dashed,thick] (10) -- ($ (3) !.5! (4) $);

\node(11)[right of =1,xshift=3cm]{8};
\draw[orange,->,very thick,snake=snake] (11)--node[midway,black,yshift=2mm]{1}(1);
\path [orange,->,draw,dashed,thick] (7) -- ($ (11) !.5! (1) $);


\node(3)[left of=1,xshift=-5cm,yshift=1cm]{$2$};
\node[teal,left of=3,xshift=1.2cm]{$\frac{2}{1+2\delta}$};
\node(4)[right of=3]{$3$};
\node(5)[right of=4]{4};
\draw[blue,->,very thick] (4)--node[midway,black,yshift=2mm]{1}(5);
\draw[orange,->,very thick,snake=snake] (3)--node[midway,black,yshift=4mm,xshift=3mm]{$\frac{2}{1+2\delta}$}(4);
\path [orange,->,draw,dashed,thick] (1) -- ($ (3) !.5! (4) $);

\node(6)[below of =3]{5};
\node(7)[right of=6]{6};
\draw[orange,->,very thick,snake=snake] (6)--node[midway,black,yshift=4mm,xshift=-3mm]{$\frac{1+2\delta}{4\delta}$}(7);
\node[teal,left of=6,xshift=1.2cm]{$\frac{1+2\delta}{4\delta}$};
\path [orange,->,draw,dashed,thick] (4) -- ($ (6) !.5! (7) $);

\node(130)[below of= 12]{8};
\node[teal,left of=130,xshift=1.2cm]{$\frac{2}{1+2\delta}$};
\node(14)[right of=130]{9};
\node(15)[right of=14]{10};
\draw[blue,->,very thick] (14)--node[midway,black,yshift=2mm]{1}(15);
\draw[orange,->,very thick,snake=snake] (130)--node[midway,black,yshift=4mm,xshift=3mm]{$\frac{2}{1+2\delta}$}(14);
\path [orange,->,draw,dashed,thick] (12) -- ($ (130) !.5! (14) $);

\node(16)[below of =130]{12};
\node(17)[right of=16]{11};
\draw[orange,->,very thick,snake=snake] (17)--node[midway,black,yshift=4mm,xshift=-3mm]{$\frac{1+2\delta}{4\delta}$}(16);
\node[teal,right of=17,xshift=-1cm]{$\frac{1+2\delta}{4\delta}$};
\path [orange,->,draw,dashed,thick] (14) -- ($ (16) !.5! (17) $);

\node(20)[left of=16,xshift=-1cm]{$b_y$};
\draw[blue,->,very thick](7)--node[midway,black,xshift=-2mm]{1}(20);
\draw[blue,->,very thick](16)--node[midway,black,yshift=2mm]{1}(20);

\node(21)[left of=20]{$13$};
\draw[blue,->,very thick](20)--node[midway,black,yshift=2mm]{1}(21);
\end{tikzpicture}
}
\caption{The financial network $\mathcal{F}_I$ that simulates instance $I$ of Figure \ref{purecircuitinstance_I}.}
    \label{purecircuitnetwork}
    \end{figure}
    
\section{Omitted Proofs }\label{apx:C}
\begin{proof}[State of the art bound.]
In \cite{schuldenzucker2017finding}, it is established that there exists a small constant $\epsilon$ for which computing an $\epsilon$-\emph{CRRV} is $\textsf{PPAD}$-\emph{hard}. To prove that, the authors defined a variation of $\epsilon$-\emph{Generalised-Circuit} (definition 5.4~\cite{schuldenzucker2017finding}) and showed that there exists an $\epsilon>0$ s.t computing $\epsilon$ approximate solutions to circuits $C'$ of this variation is $\textsf{PPAD}$-hard~(Lemma 5.6 \cite{schuldenzucker2017finding}). Subsequently they connect their version with the original used in \cite{rubinstein2015inapproximability} and show that $\epsilon/5$ solutions to circuits $C'$ of their version correspond to $\epsilon$ solutions to circuits $C$ of the original version~(Lemma 5.6 \cite{schuldenzucker2017finding}). From Theorem 4.1~\cite{DBLP:conf/focs/DeligkasFHM22} it is established that for the original version of the problem used in \cite{rubinstein2015inapproximability}, $\epsilon$-\emph{Generalised-Circuit} is $\textsf{PPAD}$-hard for $\epsilon<0.1$. Combined these facts imply that the version presented in \cite{schuldenzucker2017finding} must be inapproximable up to any factor less than 1/50. Finally it is proved that $\epsilon/3$-\emph{CRRVes} correspond to $\epsilon$ solutions for circuits $C'$ of the new version of $\epsilon$-\emph{Generalised-Circuit}~(Theorem 5.1 \cite{schuldenzucker2017finding}), which combined with the inapproximability bound implies that computing $\epsilon$-\emph{CRRV}es is \textsf{PPAD}-hard for $\epsilon<1/150$.
\end{proof}

\begin{proof}[Lemma \ref{lemma1}]
 For each substitution $\sigma$ and a term $\tau$, we will prove that $(\sigma(\tau))^\textbf{I,A} = \tau^{\textbf{I,A}}$, for any assignment \textbf{A}. Namely the interpretation dictated by \textbf{I} of the terms connected via any of the above substitutions represent the same set of values from the domain \textbf{D}. Notice that assignments where both terms are mapped to $\emptyset$, are trivially satisfied and, as such, will not be explicitly addressed in the subsequent analysis.
 \begin{enumerate}
 \item[$\sigma_{1-}$:] Let $\textbf{k} = 1-\textbf{x}$. Combining the hypothesis with the definition of operation \ref{interval substraction} we get that $\textbf{k}$ represents a range of values within $[0,1]$ and $\sup\{\textbf{k}\} = 1-\inf\{\textbf{x}\}$ and $\inf\{\textbf{k}\} = 1-\sup\{\textbf{x}\}$. If $x\in[0,1]$ then 
 \begin{enumerate}
 \item $\sup\{\textbf{k}\} = 1-\max(0,x-\epsilon) = \min(1,(1-x)+\epsilon)$. 
 \item $\inf\{\textbf{k}\} = 1-\min(1,x+\epsilon) = \max(0,(1-x)-\epsilon)$.
 \end{enumerate}
 In the case where $x>1$ then $1-x< 0$ and
 \begin{enumerate}
 \item $\sup\{\textbf{k}\} = 1-\max(0,(1-\epsilon)) = \min(1,\epsilon)$
 \item $\inf\{\textbf{k}\} = 1-\min(1,1+\epsilon) =0$ 
 \end{enumerate}
 And if $x<0$ then $1-x>1$ meaning that
 \begin{enumerate}
 \item $\sup\{\textbf{k}\} = 1-\max(0,1-\epsilon) = 1$.
 \item $\inf\{\textbf{k}\} = 1-\min(1,\epsilon) = \max(0,1-\epsilon)$.
 \end{enumerate}
 From the above analysis we conclude that,
 \begin{equation*}
 1-^\textbf{I}\textbf{x} =
 \begin{cases}
 
 \vspace{1mm}
  \Bigl[(1-x)-\epsilon,(1-x)+\epsilon\Bigr]\cap\Bigl[0,1\Bigr],&\text{ if } 1-x\in[0,1]\\

  \vspace{1mm}
  \Bigl[1-\epsilon,1\Bigr]\cap\Bigl[0,1\Bigr],&\text{ if }1-x>1\\
   
   \vspace{1mm}
\Bigl[0,\epsilon\Bigr]\cap\Bigl[0,1\Bigr],&\text{ if }1-x<0
 \end{cases}
 \end{equation*}
 which proves the claim that if $\textbf{x} = x\pm\epsilon$ then $\Bigl(1-\textbf{x}\Bigr)^\textbf{I,A} = \Bigl((1-x)\pm\epsilon\Bigr)^\textbf{I,A}$, for any assignment \textbf{A}. Consequently $\sigma_{1-}(1-\textbf{x}) = (1-x)\pm\epsilon$ is a \emph{valid substitution}.
 \item[$\sigma_{\pm}$:] Let $\textbf{k} = \textbf{x} \pm \epsilon_2$. Combining the given hypothesis with the definition of the operation $\ref{pm operation on interval}$, it is evident that the variable $\textbf{k}$ must represent a range of values within the interval $[0, 1]$. According to \ref{pm operation on interval}, we can determine that $\sup\{\textbf{k}\} = \min(1, \sup\{\textbf{x}\} + \epsilon_2)$ and
$\inf\{\textbf{k}\} = \max(0, \inf\{\textbf{x}\} - \epsilon_2)$. Furthermore, considering the operation described in \ref{pm operation on numbers}, if $x\in[0,1]$, then:
\begin{enumerate}
\item $\sup\{\textbf{k}\} = \min(1, x + \epsilon_1 + \epsilon_2)$.
\item $\inf\{\textbf{k}\} = \max(0, x - \epsilon_1 - \epsilon_2)$.
\end{enumerate}
In cases where $x > 1$, then $\sup\{\textbf{x}\} = 1$ and $\inf\{\textbf{x}\} = \max(0, 1 - \epsilon_1)$ which implies that:
\begin{enumerate}
\item $\sup\{\textbf{k}\} = 1$.
\item $\inf\{\textbf{k}\} = \max(0, 1 - \epsilon_1 - \epsilon_2)$.
\end{enumerate}
And when $x < 0$, it holds that:
$\sup\{\textbf{x}\} = \min(1, \epsilon_1)$ and $\inf\{\textbf{x}\} = 0$, which leads to the conclusion that
\begin{enumerate}
\item $\sup\{\textbf{k}\} = \min(1, \epsilon_1 + \epsilon_2)$.
\item $\inf\{\textbf{x}\} = 0$.
\end{enumerate}
 In summary of the preceding analysis, we can conclude that,
 \begin{equation*}
 \textbf{x}\pm^\textbf{I}\epsilon_2 =
 \begin{cases}
 
 \vspace{1mm}
  \Bigl[x-(\epsilon_1+\epsilon_2),x+(\epsilon_1+\epsilon_2)\Bigr]\cap\Bigl[0,1\Bigr],&\text{ if } x\in[0,1]\\

  \vspace{1mm}
  \Bigl[1-(\epsilon_1+\epsilon_2),1\Bigr]\cap\Bigl[0,1\Bigr],&\text{ if }x>1\\
   
   \vspace{1mm}
\Bigl[0,\epsilon_1+\epsilon_2\Bigr]\cap\Bigl[0,1\Bigr],&\text{ if }x<0
 \end{cases}
 \end{equation*}
 which proves the claim that if $\textbf{x}=x\pm\epsilon_1$ then $\Bigl(\textbf{x}\pm\epsilon_2\Bigr)^{\textbf{I,A}} = \Bigl(x\pm(\epsilon_1+\epsilon_2)\Bigr)^\textbf{I,A}$. Consequently $\sigma_{\pm}(\textbf{x}\pm\epsilon_2)=x\pm(\epsilon_1+\epsilon_2)$ is a valid substitution.
 
 \item[$\sigma_{+}$:] Consider $\textbf{k} = \textbf{x} + \textbf{y}$. It's essential to note that, as per the definition in operation \ref{pm operation on numbers}, both variables $\textbf{x}$ and $\textbf{y}$ represent intervals constrained to the range $[0,1]$. This, in conjunction with operation \ref{interval addition}, underscores that the addition of intervals is exclusively defined within the boundaries of $[0,1]$. 
 
 In line with operation \ref{interval addition}, when $x + y\in [0,1]$, we observe that
 \begin{enumerate}
 \item $\sup\{\textbf{k}\} = \min(1, \sup\{\textbf{x}\} + \sup\{\textbf{y}\})$ 
 \item $\inf\{\textbf{k}\} = \max(0, \inf\{\textbf{x}\} + \inf\{\textbf{y}\})$.
 \end{enumerate}

Whenever $x + y > 1$, then in line with the definition provided in operation \ref{interval addition}, it holds that:
\begin{enumerate}
\item $\sup\{\textbf{k}\} = 1$
\item $\inf\{\textbf{k}\} = \max(0, 1 - (\epsilon_1 + \epsilon_2))$
\end{enumerate}
Finally whenever $x + y < 1$:
\begin{enumerate}
\item $\sup\{\textbf{k}\} = \min(1, \epsilon_1 + \epsilon_2)$.
\item $\inf\{\textbf{k}\} = 0$.
\end{enumerate}
Directly from \ref{interval addition} and by the above analysis we get that, 
\begin{equation*}
 \textbf{x} +^\textbf{I} \textbf{y} =
 \begin{cases}
 
 \vspace{1mm}
  \Bigl[(x+y)-(\epsilon_1+\epsilon_2),(x+y)+(\epsilon_1+\epsilon_2)\Bigr]\cap\Bigl[0,1\Bigr],&\text{ if } x+y\in[0,1]\\

  \vspace{1mm}
  \Bigl[1-(\epsilon_1+\epsilon_2),1\Bigr]\cap\Bigl[0,1\Bigr],&\text{ if }x>1\\
   
   \vspace{1mm}
\Bigl[0,\epsilon_1+\epsilon_2\Bigr]\cap\Bigl[0,1\Bigr],&\text{ if }x<0
 \end{cases}
 \end{equation*} 
which establishes that $\Bigl(\textbf{x}+ \textbf{y}\Bigr)^\textbf{I,A} = \Bigl((x+y)\pm(\epsilon_1+\epsilon_2)\Bigr)^\textbf{I,A}$. This implies that whenever $\textbf{x} = x\pm\epsilon_1$ and $\textbf{y} = y\pm\epsilon_2$ then $\sigma_{+}(\textbf{x}+\textbf{y}) = (x+y)\pm(\epsilon_1+\epsilon_2)$ is a valid substitution.

\item[$\sigma_{*}$:] Assume that $\textbf{k} = l \cdot \textbf{x}$. The claim naturally stems out from the definition of operation \ref{number interval mult}. If $l \cdot x\in[0,1]$, we have the following:
\begin{enumerate}
\item $\sup\{\textbf{k}\} = \min(1, l \cdot \sup\{\textbf{x}\})$.
\item $\inf\{\textbf{k}\} = \max(0, l \cdot \inf\{\textbf{x}\})$.
\end{enumerate}

In situations where $l \cdot x > 1$, as defined in operation \ref{number interval mult}:
\begin{enumerate}
\item $\sup\{\textbf{k}\} = 1$
\item $\inf\{\textbf{k}\} = \max(0, 1 - l \cdot \epsilon)$
\end{enumerate}
Conversely, if $l \cdot x < 0$, we have:
\begin{enumerate}
\item $\sup\{\textbf{k}\} = \min(1, l \cdot \epsilon)$
\item $\inf\{\textbf{k}\} = 0$
\end{enumerate}
 Summing up we get that:
 \begin{equation*}
 l\cdot^\textbf{I}\textbf{x} =
 \begin{cases}
 
 \vspace{1mm}
  \Bigl[l\cdot x-l\cdot\epsilon,l\cdot x+l\cdot\epsilon\Bigr]\cap\Bigl[0,1\Bigr],&\text{ if }l\cdot x\in[0,1]\\

  \vspace{1mm}
  \Bigl[1-l\cdot\epsilon,1\Bigr]\cap\Bigl[0,1\Bigr],&\text{ if }l\cdot x>1\\
   
   \vspace{1mm}
\Bigl[0,l\cdot\epsilon\Bigr]\cap\Bigl[0,1\Bigr],&\text{ if }l\cdot x<0
 \end{cases}
 \end{equation*} 
 which establishes that if $\textbf{x} = x\pm\epsilon$ and $l$ represents a number then $\Bigl(l\cdot\textbf{x}\Bigr)^\textbf{I,A} = \Bigl(l\cdot x\pm l\cdot\epsilon\Bigr)^\textbf{I,A}$. This implies that $\sigma_{*}(l\cdot\textbf{x}) = l\cdot x\pm l \cdot\epsilon$ is a valid substitution.
\end{enumerate}
\end{proof}

\begin{proof}[Lemma \ref{lemma2}]

Given the assumptions, we will prove that the pair of values $(\textbf{y}^{\textbf{I,A}}, \textbf{x}^{\textbf{I,A}})$ are contained in the relation \ref{relation}.

When $x \in [0, 1]$ and  $x \geq y$ it is clear that $y \leq 1$. This leads to the following comparisons:
\begin{enumerate}
\item $\sup\{\textbf{x}\} = \min(1, x + \epsilon) \geq \min(1, y + \epsilon) = \sup\{\textbf{y}\}$.
\item $\inf\{\textbf{x}\} = \max(0, x - \epsilon) \geq \max(0, y - \epsilon) = \inf\{\textbf{y}\}$.
\end{enumerate}
Now, when $x > 1$ and $x \geq y$, we have:
\begin{enumerate}
\item $\sup\{\textbf{x}\} = 1 \geq \sup\{\textbf{y}\}$.
\item $\inf\{\textbf{x}\} = \max(0, 1 - \epsilon)$.
If $y > 1$, then $\inf\{\textbf{y}\} = \max(0, 1 - \epsilon)$, while otherwise $\inf\{\textbf{y}\} = \max(0, y - \epsilon)$. In both cases, $\inf\{\textbf{x}\} \geq \inf\{\textbf{y}\}$.
\end{enumerate}
Finally, if $y \leq x < 0$, according to \ref{pm operation on numbers}, we have $\sup\{\textbf{x}\} = \sup\{\textbf{y}\}$ and
$\inf\{\textbf{x}\} = \inf\{\textbf{y}\}$.
In all scenarios, it is evident that $\sup\{\textbf{x}\} \geq \sup\{\textbf{y}\}$ and $\inf\{\textbf{x}\} \geq \inf\{\textbf{y}\}$, which, by definition, implies that $(\textbf{y}^\textbf{I,A},\textbf{x}^\textbf{I,A})\in \preceq^\textbf{I}$. 

\end{proof}

\end{document}